\documentclass[onecolumn,11pt]{article}
\usepackage{amsmath,amssymb}
\usepackage{amsthm,amsmath}
\usepackage{makeidx,graphicx}
\usepackage{algorithmic}
\usepackage{algorithm,xcolor}
\usepackage{tkz-berge}
\usetikzlibrary{fit,shapes}
\def\ind{{\mathchoice {\rm 1\mskip-4mu l} {\rm 1\mskip-4mu l}
{\rm 1\mskip-4.5mu l} {\rm 1\mskip-5mu l}}}

\newtheorem{proposition}{Proposition}
\newtheorem{defi}{Definition}

\newtheorem{lem}{Lemma}

\begin{document}

\title{Lowest Unique Bid Auctions \\ with Resubmission Opportunities}

\author{Yida Xu and Hamidou Tembine \thanks{ The authors are with Learning \& Game Theory Lab, New York University Abu Dhabi, Email: \{yida.xu,tembine\}@nyu.edu} }

\maketitle

\begin{abstract}
The recent online platforms propose multiple
items for bidding. The state of the art, however, is limited to the analysis of one item auction without resubmission.
In this paper we study multi-item lowest unique bid auctions (LUBA)  with resubmission in discrete  bid spaces under budget constraints.  We show that the game does not have pure Bayes-Nash equilibria
(except in very special cases). However,  at least one mixed Bayes-Nash equilibria exists for arbitrary number of bidders and items.
The equilibrium is explicitly computed for two-bidder setup with resubmission possibilities. In the general setting we propose a  distributed strategic learning algorithm  to
approximate equilibria. Computer simulations indicate that the error quickly decays in few number of steps. When the number of bidders per item
follows a Poisson distribution, it is shown that the seller can get a non-negligible revenue on several items, and hence making a partial revelation of the true value of the items.
Finally,  the attitude of the bidders towards the risk is considered. In contrast to risk-neutral agents who bids very small values, the cumulative distribution and  the bidding support of  risk-sensitive agents are more distributed.

\end{abstract}


\newpage

\tableofcontents

\newpage

\section{Introduction}
Information technology has revolutionized the traditional structure of economic and financial markets.
The removal of geographical and time constraints has fostered the growth of online auction markets, which now include millions of economic agents and companies worldwide
and annual transaction volumes in the billions of US dollars. Here, we study bidders' learning and behavior of a little studied type of online auctions  called
lowest unique bid auction (LUBA). LUBAs are online auctions which have reached a considerable success during last decade.
Their key feature is that they are reverse auctions: rather than the bidder with the highest bid (as in the case of traditional auctions), the winner is the bidder who makes the lowest unique bid.
The recent online platforms propose even multiple items and bidders can submit their bidders over several rounds [before a winner is decided].
The bidding status  that a bidder observes changes according to the actions (bids) chosen by all the (active) bidders on  the corresponding item.
 Due to limited number of items and budget restrictions, bidders need to manage their decision in a  strategic way.

\subsection*{Literature review}

{\bf One-item Auctions:} The theory of auctions as games of incomplete information originated in 1961 in the work of Vickrey~\cite{reftv3}.
 While  auctions with homogeneous valuation distributions (symmetric auctions) and self-interested non-spiteful bidders is well-investigated in the literature, auction with asymmetric bidders remain a challenging open problem~(see \cite{reftv1,lebrun1,lebrun2} and the references therein).
 With asymmetric auctions, the  expected ``revenue equivalence theorem'' \cite{myeron1981} does not hold, i.e., there is a class of cumulative valuation distribution  function such that the revenue of the seller (auctioneer) depends on the auction mechanism employed. In addition, there is no ranking revenue between the auction mechanisms (first, second, English or Deutch).

{\bf Multi-Item Auctions:}
There are few research articles on multi-item auctions \cite{multiitem1,multiitem2,multiitem3,multiitem4}. Most of these works present
computer simulation  and numerical experiments results. However, no analysis of the outcome of the multi-item auction is available. There is no analysis  of the   equilibrium seeking algorithm  therein.
\begin{figure}[!htb]
  \centering
   \includegraphics[width=0.97\textwidth, height=7cm]{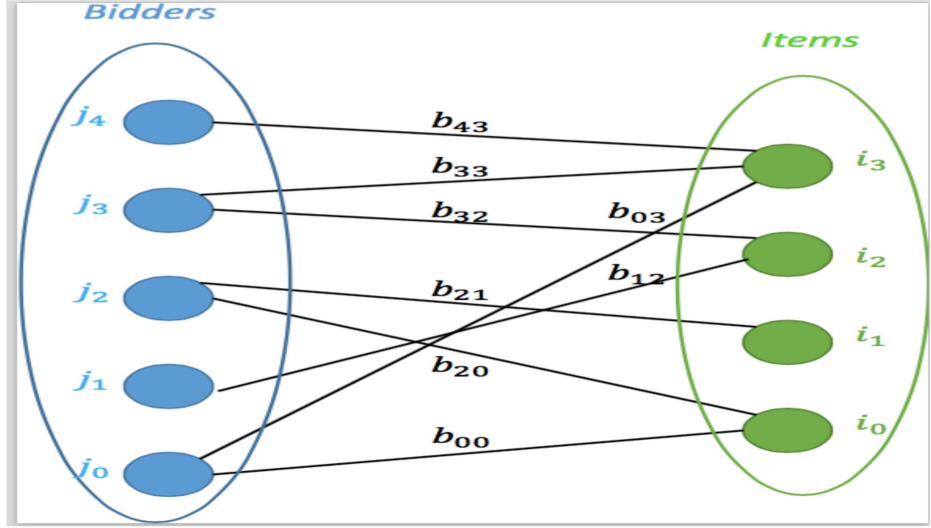}
  \caption{Depending its remaining budgets, each bidder may act on multiple items and may place multiple bids per item} \label{fig1}
\end{figure}
The above mentioned works do not consider the the resubmission feature. Note however that  the possibility for a bidder to resubmit another bid for the same item is already implemented in practice in the online auction markets.

{\bf LUBA:}  LUBA is very different than the second price auctions. The particularity of LUBA is that it focuses on the lowest unique bid, which creates lot of difficulties in terms of analysis. LUBA is different that the lowest cost auction called procurement auction which is widely used in e-commerce and cloud resource pooling \cite{infocom15v1} or in demand-supply matching in power grids \cite{infocom15v2}.
Single-item LUBAs are a special case of unmatched bid auctions which have been studied by other researchers \cite{houba,dwn,rapp,wang,vino0,nicolas}.  The authors \cite{dwn} run laboratory experiments with minbid auctions. They consider the case where players are restricted to only one bid and compare the results from their laboratory experiment with a Monte Carlo simulation.  The authors in \cite{rapp} consider high and low  unique bid auctions where bidders are also restricted to a single bid. They provide a numerical approximation of the solution for a game-theoretic model and compare it with the results of a laboratory experiment. The work in  \cite{wang} conducts a lowest unique positive integer experiment and contrast the observed behavior with the solution of a Poisson game with a single bid per player. \cite{vino0,vino1,vino2} conducted  field experiments on Lowest-Unmatched Price Auctions with mostly large prizes involving large numbers of participants (tens and hundreds of thousands).
\cite{huang}  studies truthful multi-unit transportation procurement auctions. The work in \cite{ecom} discusses security and privacy issues of multi-item reverse Vickrey auction by designing  more secure protocols.
Most of the above works restrict the number of submissions per bidder to one. The recent focus within the auction field has been multi-item auctions where bidders are not restricted to buying only one item of the merchandise. The bidder can also place multiple bids for each item. It has been of practical importance in Internet auction sites and has been widely executed by them.

\subsection*{Contribution}
Our contribution can be summarized as follows. Mimicking online platform auctions, we propose and analyze a multi-item LUBA game with budget constraint, registration fee and resubmission cost. We show that the analysis can be reduced into a constrained finite game (with incomplete information) by eliminating the bids that are higher than the value of the item. Using classical fixed-point theorems, there is at least one Bayes-Nash equilibrium in mixed strategies.  Next, we address the question of computation and stability of such an equilibrium. We provide explicitly the equilibrium structure in special cases.
We provide a learning algorithm that is able to locate equilibria.
An imitative combined fully distributed payoff and strategy learning (imitative CODIPAS learning) that is adapted to LUBA is proposed to locate/approximate equilibria. We examine how the bidders of the game  are able to learn about the online system output using their own-independent learning strategies and own-independent valuation. The numerical investigation  shows that the proposed algorithm can effectively learn Nash equilibrium  in few  steps.  It is shown that the auctioneers can make a positive revenue when the number of bidders per bid exceeds a certain threshold.
We then examine the attitude of the bidders towards the risk. In contrast to risk-neutral agents who bids very small values, the cumulative distribution and  the bidding support of  risk-sensitive agents are more distributed.

\subsection*{Structure of the paper}
The paper is organized as follows. In Section \ref{sec:2}, we introduce LUBA mechanism. In Sect \ref{sec:3} we set up the problem statement, whose solution approach is given in Section \ref{sec:4}. We provide an imitative learning algorithm for approximating Nash equilibria in Section \ref{sec:5}. Section \ref{sec:2risk} focuses on risk-sensitive bidders' behaviors.
Section \ref{sec:9} concludes the paper.

We summarize some of the notations in Table \ref{tablenotationjournal}.
\begin{table}[htb]
\caption{Summary of Notations} \label{tablenotationjournal}
\begin{center}
\begin{tabular}{ll}
\hline
  Symbol & Meaning \\ \hline
  $\mathcal{J}$ & set of potential bidders  \\
  $n$ & cardinality of $\mathcal{J}$\\
  $\mathcal{B}$ & bid (action) space\\
  $\mathcal{I}$ & set of items (from auctioneers) \\
  $m$ & cardinality of $\mathcal{I}$\\
  $B_{ji}$ & bid set of bidder $j$ on item $i$ \\
  $r_0(B_1,\ldots,B_n)$ & payoff of the auctioneer\\
  $r_j(B_1,\ldots,B_n | \ v_j)$ & payoff of  bidder $j$\\
  $\ind_{\{.\}}$ & indicator function.\\
  $F$ & cumulative distribution of the valuation matrix\\
  $F_j$ & (marginal) cumulative distribution of \\ & the valuation vector of bidder $j$\
  \\ \hline
\end{tabular}
\end{center}
\end{table}

\section{Background on LUBA} \label{sec:2}
A lowest unique bid auction operates under the following three main rules:
\begin{itemize}
\item Whoever bids the lowest unique positive amount wins.
\item If there is no unique bid then no one wins. In particular, the bid of the winner should be unmatched.
\item No participant can see bids placed by other participants.
\end{itemize}

The term "lowest unique bid'' means  is the lowest amount that nobody else has bid and  its computation is illustrated in Table \ref{ILUBATable} and Fig.\ref{ILUBAFigure} below.

\begin{table}[htb]
\caption{Unique lowest bid auctions} \label{ILUBATable}
\begin{center}
\begin{tabular}{|c|c|c|}
\hline
Bid Amount & Number of Bids & Status\\
\hline
1 cent & 2 & Not unique \\
\hline
2 cents & 0 &  \\
\hline
3 cents & 1 & Lowest Unique Bid!  \\
\hline
4 cents & 3 & Not Unique\\
\hline
5 cents & 2 &  \\
\hline
6 cents & 1 & Unique but not the Lowest\\
\hline
7 cents & 0 &  \\
\hline
8 cents & 1 &  Highest unique bid\\
\hline
9 cents & 3 & Not unique, highest bid \\
\hline
10 cents & 0 & \\
\hline
\end{tabular}
\end{center}
\end{table}
\begin{figure}[!htb]
  \centering
    \includegraphics[width=0.9\textwidth]{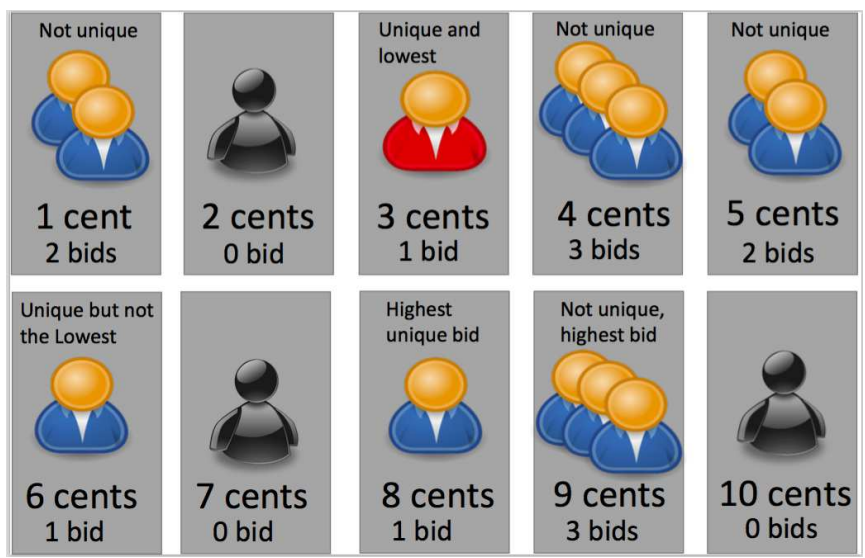}
  \caption{Unique lowest bid auctions. Illustration of the rules of a LUBA system. The winner results to be the bidder who has bid 3 cents, which is the lowest unique bid. All the other bids are not unique expect from the one of 6 cents and 8 cents, however, they are not the lowest one. By contrast, in highest unique bid auctions the rules are reversed, and the winner is the bidder making the unique highest bid (8 cents)} \label{ILUBAFigure}
\end{figure}

 Bids may be any amount between $0.01$ and $10^3$ allowing people to buy an item at an incredibly low price. The cost of the item is covered by the entry/administration/bid fee  paid by  all participants and by the resubmission fee to make a bid.

 In this example the Bid of 3 cents is the Lowest Unique Bid and is awarded the auction result and would have to pay only 3 cents after the registration fee. He/she has therefore purchased an item for 3 cents + the Administration fee paid when he/she placed the bid and the registration fee.
If the Bidder of the successful bid had more than one bid they would need to add a certain cost for each additional bid placed.

Each bidder can decide to participate or not. A participant can start bidding in just few steps:

\begin{itemize} \item Register - complete the registration form. Your password will then be emailed to the email address you provide us. It is vital you keep your username and password confidential, otherwise you may allow others to access your Bid Bank and place bids from your account.

The user is fully responsible for all activities that occur through your subscription and under your password and account.
\item  Login - once the user receive its password via email, he/she login using her email address as your username and your password. For security issue the user can change her password.

\item Purchase Credit - Buy Credits from your account. Smiles or other related credits can be converted into a auction credit.

\end{itemize}

\section{Problem Statement} \label{sec:3}

\subsection*{Multi-item contest} A multi-item contest is a situation in which players exert effort for each item in an attempt to win a prize.
Decision to participate or not is costly and left to the players. All the efforts are sunk while only
the winner gets the prize.
An important ingredient in describing a multi-item contest  with the set of potential participants $\mathcal{J}=
\{1, . . . , n\}$, the set of auctioneers proposing the set of items $\mathcal{I}=\{1,\ldots, m\}$
 is the contest success function, which takes the efforts $B$  of the agents and
converts them into each agent's probability of winning per item: $$P_j : \ {B}=(B_{ji})_{j,i}\subset \mathbb{N}^{n\times m} \mapsto [0, 1]^m. $$ The (expected)
payoff of a risk-neutral player $j$ with  $v_{ji}$ a value of winning  item $i$ and a cost of effort function $c_{ji}$ is
 $$\sum_{i\in \mathcal{I}}P_{ji}(B) v_{ji} -c_{ji}(B).$$

\subsection*{Multi-item LUBA}One very popular multi-item contest used in online platform is multi-item LUBA. In multi-item LUBA,
Multiple sellers (auctioneers) have multiple items (objects)   to sell on the online platform. These sellers have adopted a lowest unique bid auction (LUBA) rule.
 The multi-item LUBA game (with resubmission) is as follows.

 There are $n\geq 2$ bidders for $m$ items proposed by the sellers.
 A bidder's assessment of the worth of the offered object for auction is called a value.
 In an multi-item auction context a bidder has a vector of values, one  value per good.

 {\bf Incomplete information about the others: } A bidder may have its own valuation  vector but not the valuation vector of the others.

 The bidders are assumed to have  (possibly heterogeneous) valuation distributions.
 Each bidder independently submits a possibly several bid per item without seeing the others' bids. The submission fee per bid on item $i$ is $c_{i}.$
 If there is only a unique lowest bid, the object is sold to the bidder with unique lowest bid.
 Each bidder pays the cost $c_r$ for the registration and administration fee. In addition, the winner pays her winning bid on item $i$, that is, the price is  the lowest unique bid on that item.
  Note that the existence of a winning bid is not guaranteed, as for example, $n$ identical bids demonstrate. In the absence of a winner, the item remains with the auctioneer. Note that, a tie-breaking rule can be used in that case. We denote by $v_{ji}$ the valuation of bidder $j$ for item $i.$ The random variable ${v}_{ji}$ has  support $[\underline{v},\bar{v}]$ where  $0<\underline{v}< \bar{v}.$ Each bidder $j$ has a initial total budget of $\bar{b}_j$ to be used for all items. Each bidder $j$ knows its own-valuation  vector $v_j=(v_{ji})_{i\in \mathcal{I}}$ and own-bid  vector $(b_{ji})_{i\in \mathcal{I}}$ but not  $v_{-j}=(v_{j'})_{j'\neq j}$ the valuation of the other bidders.
Note that each bidder can resubmit bids a certain number of times subject to her available budget, each resubmission for
 item $i$ will cost $c_i.$ If bidder $j$
 has (re)submitted $n_{ji}$ times on item $i$ her total submission/bidding cost would be $n_{ji}c_i$ in addition to the registration fee. Denote
 the set that contains all the bids of bidder $j$ on item $i$ by  $B_{ji}\subset \mathbb{N}.$  Thus, $n_{ji}=|B_{ji}|$ is the cardinality of the strictly positive bids by $j$ on item $i.$ The set of bidders who are submitting $b$ on item $i$ is denoted by
 $$ N_{i,b}=\{ j\in \mathcal{J} \ | \ b \in B_{ji} \}. $$ In order to get the set of all  unique bids, we introduce the following:
 The set of all positive natural numbers that were chosen by only one bidder on item $i$ is
 $$B_{i}^*=\{ b>0 \ |  \  | N_{i,b} |=1 \}.$$
 If $B_{i}^*=\emptyset$ then there is no winner on item $i$ at that round (after all the resubmission possibilities). If $B_{i}^*\neq \emptyset$ then there is a
 winner on item $i$ and the winning bid is $\inf B_{i}^*$ and winner is $j^*\in N_{i,\inf B_{i}^*}.$ The payoff of bidder $j$ on item $i$ at that round would be
 $$
 r_{ji}=v_{ji}-|B_{ji}| c_i - \inf B_{i}^*-c_r,
 $$ if $j$ is a winner on item $i$, and
 
 $$
 r_{ji}=-|B_{ji}| c_i - c_r,
 $$ if $j$ is not a winner on item $i.$   The payoff of bidder $j$ on item $i$ is zero if $B_{ji}$ is reduced to $\{0\}$ (or equivalently the empty set).

 \begin{eqnarray}&& r_{ji} (B) \\  \nonumber
 &=& [-c_r-c_i |B_{ji}|-(v_{ji}-b_{ji})\ind_{\{ b_{ji}=\inf B_i^*\}}]\ind_{\{  B_{ji}\neq \{0\}\}},
 \end{eqnarray}
where the infininum of the empty set is zero.

  \begin{eqnarray} r_j (B) =\sum_{i\in \mathcal{I}} r_{ji}(B).
 \end{eqnarray}

 The instant payoff of the auctioneer of item $i$ is $$ r_{a,i}=\left(\sum_{j}c_r\ind_{\{ B_{ji}\neq \emptyset \}}+ \inf B_{i}^*+\sum_{j=1} |B_{ji}| c_i \right)  -  v_{a,i},$$ where $v_{a,i}$ is the realized valuation of the auctioneer for item $i.$
 The instant payoff of the auctioneer of a set of item $I$ is $ r_{a,\mathcal{I}}=\sum_{i\in \mathcal{I}} r_{a,i}.$
Bidders are interested in optimizing their payoffs and the auctioneers are interested in their revenue.
\subsection{Solution Concepts}
Since the game is of incomplete information, the strategies must be specified as a function of the information structure.
 \begin{defi}
 A pure strategy of a bidder  is a choice of a subset of natural numbers given the own-value and own-budget.
Thus, given its own valuation vector
 $v_j=(v_{ji})_i,$ bidder $j$ will choose an action  $(B_{ji})_i$  that satisfies the budget constraints $$\  \sum_{i}c_r\ind_{\{ B_{ji}\neq \{0\} \}}+\sum_{i=1}^m  [\inf B_{i}^*]\ind_{B_{ji} \cap [\inf B_{i}^*]}+\sum_{i=1}^m|B_{ji}|c_i\leq \bar{b}_j.$$  The set of multi-item bid space for bidder $j$ is
 $$\begin{array}{c}
 \mathcal{B}_j(v_j,\bar{b}_j)=\{ (B_{ji})_i \ | \  \  B_{ji} \subset \{0,1,\ldots, \bar{b}_j-c_r\},\  \\  \sum_{i}c_r\ind_{\{ B_{ji}\neq \{0\} \}}+\sum_{i=1}^m  [\inf B_{i}^*]\ind_{B_{ji} \cap [\inf B_{i}^*]}+\sum_{i=1}^m|B_{ji}|c_i\leq \bar{b}_j\}.
 \end{array}$$

 A pure strategy is a mapping $ v_j \mapsto B_j\subset \mathbb{N}.$
  A  constrained pure strategy is a mapping $ v_j \mapsto B_j\in \mathcal{B}_j.$ A mixed strategy is a probability measure over the set of pure strategies.

  \end{defi}

   The action set $ \mathcal{B}_j(v_j,\bar{b}_j)$ is finite because of budget limitation.
 A bid $b_{ji}\in B_{ji}$ is hence less than $\min(\bar{b}_j, v_{ji}-c_r).$

%

\subsection{Bidders' equilibria }

We define a solution concept of the above game with incomplete information: Bayes-Nash equilibrium.

\begin{defi}
A mixed Bayes-Nash strategy equilibrium is a profile $(s_{j}(v_j))_j$ such that for all bidders $j$
$$
\mathbb{E}_{s_j,s_{-j}}{r}_{j}(B_j(v_j),B_{-j} | \ v_j )\geq \mathbb{E}_{s_j',s_{-j}}{r}_{j}(B'_j,B_{-j} | \ v_j),
$$
for any strategy $s'_j.$
\end{defi}

\section{Analysis: Risk-Neutral Case} \label{sec:4}
We are interested in the equilibria, equilibrium payoffs of the bidders and revenue of the auctioneer.
Note that the information structure is significantly reduced.
Since $j$  does not know the distribution of the random matrix $v_{-j}=(v_{j'})_{j'\neq j}$ which may influence $B_{-j},$ it is unclear how can bidder $j$ evaluate the expected payoff $\mathbb{E}_{s_j,s_{-j}}{r}_{j}(B_j(v_j),B_{-j} | \ v_j ).$  Therefore the expected payoff needs to be learned by $j.$

\begin{proposition} \label{domino} Let $c_i>0.$ Any realized value of item $i$ such that $v_{ji} < c+c_r$  leads to a trivial choice (i.e., $\{0\}$)  for bidder $j.$ Any  participative bidding strategy $B_{ji}$ on item $i$ such that $|B_{ji}|> \frac{v_{ji}-c_r}{c_i} $ is dominated by the strategy $\{0\}.$  Any strategy $B$ such $b_{ji}\ind_{b_{ji}=\inf B_i^*}+c+c_r > \bar{b}_j$ is dominated by $\{0\}$ on item $i.$
\end{proposition}

\begin{proof}[Proof of Proposition \ref{domino}]
By budget constraint, $j'$s bids must fulfill  the budget restriction $$\sum_i c_r\ind_{B_{ji}\neq \{0\}}+\sum_{i}|B_{ji}| c_i + \sum_{i} b_{ji}\ind_{b_{ji}=\inf B_i^*}\leq \bar{b}_j.$$
If agent $j$ bids on item $i$ with $b_{ij}>v_{ji} - {c}_r-c$ then $j$ gets in the bets case  $v_{ji} - \tilde{c}-b_{ji}$ which is negative (loss) and $j$ could guarantee zero as payoff by not participating. Therefore the strategy $0$ dominates any $b_{ji}$ higher than $v_{ji} - \tilde{c}.$  Thus, the bid space of agent $j$ on item $i$ can be reduced to $$B_{ji}\subset  \prod_{i=1}^m \{0,1,2,\ldots,  \min(v_{ji} - \tilde{c}, \bar{b}_j)  \}$$ such that
	$$\sum_i c_r\ind_{B_{ji}\neq \{0\}}+\sum_{i}|B_{ji}| c_i + \sum_{i} b_{ji}\ind_{b_{ji}=\inf B_i^*}\leq \bar{b}_j.$$ This completes the proof.
	\end{proof}

As a consequence of Proposition \ref{domino},
\begin{itemize} \item the number of resubmissions on item $i$ needs to be bounded by  $|B_{ji}|\leq  \frac{v_{ji}-c_r}{c_i}$ to be rewarding.
\item when analyzing Bayes-Nash equilibria one can limit  the action space to all the subsets of the finite set  $\prod_{i=1}^m \{0,1,2,\ldots,  \min(v_{ji} - {c}_r, \bar{b}_j)  \}$
\end{itemize}

\begin{proposition} \label{nopure} Generically (under feasible budget), the LUBA game with resubmission has no pure Bayes-Nash equilibria.
\end{proposition}
	
\begin{proof}[Proof of Proposition \ref{nopure} ]
Let $c>0, \ n\geq 2$ and consider a pure behavioral strategy profile $B$  with a generic budget $\min(\bar{b}_j, v_{ji})>2+c+c_r.$
If $B_{ij}$ does not contain a winning bid $b_{ij}$ then $B_{ji}$ can be reduced by saving the resubmission cost, i.e.,  an empty set or $\{0\}.$ But if $j$ is not participating on item $i$ (while its budget allows) then there is another player $-j$  who could also save by decreasing its bid set. However, the action $a_0=(\{0\}, \{1\}, \ldots, )$ cannot be an equilibrium because player $j$ can deviate and bids $\{1234\ldots, k\}$ with $k$ higher than  the maximum bid in action $a_0.$ Hence, there always a player who can deviate and benefits if budget allows and if the value $v_{ji}$ is not reached in terms of bidding cost.

If $B_{ij}$ contains a winning bid $b_{ij},$  then there is another player $j'\neq j$ such that $B_{j'i}$ does not contain the winning bid and one can apply the reasoning above with $B_{j'i}.$
Iterating this for all items, we deduce that the action $B$ is not a best response to itself. Hence, $B$  cannot be a Bayes-Nash equilibrium. This completes the proof.
\end{proof}

Note however that there are trivial cases with pure equilibria:
\begin{itemize}\item if multiple resubmissions  are not allowed i.e., $|B_{ji}|\leq 1$ then the action profile $a_1=(\{0\},\ldots,\{0\},\{1\},\{0\},\ldots, \{0\})$ is an equilibrium profile of that item whenever the realized value is such that $v_{ji}>1+c_i+c_r.$
Putting together, the set of pure equilibria in this very special case game
$G(J,(\overline{b}_j)_{j\in J},(c_i)_{i\in \mathcal{I}},c_r,m,(F_j)_{j\in J})$ with  $|B_{ji}|\leq 1,$  $F_j$ is the cumulative distribution of $v_j,$
the set of action $A_j=D_0\cup D_1\cup D_2\cup \ldots\cup D_{\overline{b}_j},$ where $D_k=\{ d,\ d_1+\ldots+d_m= k, d_l\in \mathbb{N} \}$   the set of decomposition/partition of the number $k,$
is given by
\[
\begin{bmatrix}
    b_{11}       & b_{12} & b_{13} & \dots & b_{1m} \\
    b_{21}       & b_{22} & b_{23} & \dots & b_{2m} \\
    b_{31}       & b_{32} & b_{33} & \dots & b_{3m} \\
    \hdotsfor{5} \\
    b_{n1}       & b_{n2} & b_{n3} & \dots & b_{nm}
\end{bmatrix}
\]
where \begin{itemize}\item[(a)] $b_{ji}\in \{0,1\}, \ \forall (i,j)\in \mathcal{I}\times \mathcal{J}$
\item[(b)]  $\sum_{j=1}^n b_{ji}=1,\  \ \forall i\in \mathcal{I}$ and \item[(c)]
$\sum_{i=1}^m b_{ji}\leq \overline{b}_j,\ , \ \forall j\in  \mathcal{J}$\end{itemize}

\begin{figure}[!htb]
  \centering
    \includegraphics[width=0.97\textwidth]{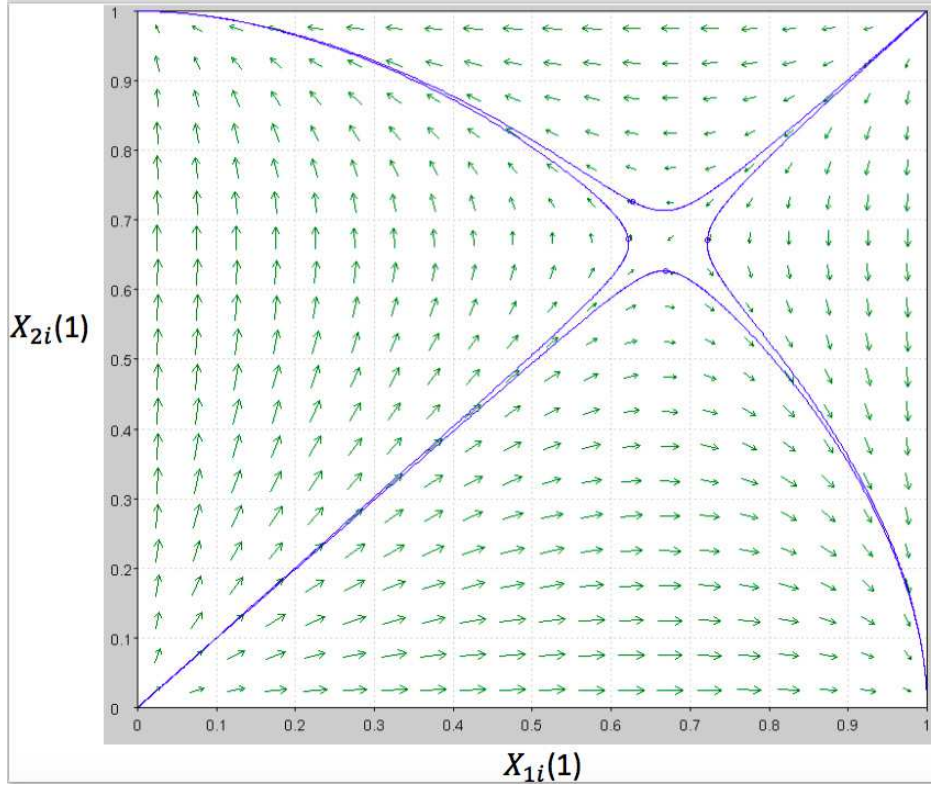}
  \caption{The vector field of the learning dynamics when $n=2,v=4,c=1.$ Convergence to equilibrium.} \label{vectorfield}
\end{figure}
 As shown in Figure \ref{vectorfield}, the learning algorithm always converges to one of pure Nash equilibria except when they are symmetric, however, the probability of symmetric situation is zero.

\item  If $v_{ji}\in [1+c+c_r, 2+c+c_r)$ it is not interesting to bid higher than $1$ because of negative payoff.
\item when the realized  value $v_{ji}$ is below $c+c_r,$ the action  $\{0\}$ is better for $j$ on $i.$
\end{itemize}

\subsubsection*{Non-potential game}
As a corollary of Proposition \ref{nopure} we deduce that LUBA is NOT a potential game because one can exhibit the following cycle of action profiles for item $i$

$(\{1\},\{0\},\ldots,\{0\})\mapsto (\{1\},\{12\},\ldots,\{0\})$

 $ \mapsto (\{1\},\{12\}, \{123\}, \{0\}\ldots,\{0\})$

$\vdots$

$\mapsto (\{1\},\{12\}, \{123\}, \{1234\}\ldots,\{1234\ldots \bar{b}\})$

 $\mapsto (\{0\},\{12\}, \{123\}, \{1234\}\ldots,\{123\ldots \bar{b}\})$

 $\mapsto (\{0\},\{0\}, \{123\}, \{1234\}\ldots,\{123\ldots \bar{b}\})$

 $\mapsto (\{0\},\{0\}, \{0\}, \{1234\}\ldots,\{123\ldots \bar{b}\})$

 $\vdots$

$\ldots  (\{0\},\{0\}, \{0\}, \{0\}\ldots, \{0\},\{123\ldots \bar{b}\})$

$\mapsto (\{0\},\{0\}, \{0\}, \{0\}\ldots, \{0\},\{1\})$

$\mapsto (\{12\},\{0\}, \{0\}, \{0\}\ldots, \{0\},\{1\})$

$\mapsto (\{12\},\{0\}, \{0\}, \{0\}\ldots, \{0\},\{0\})$

$\mapsto (\{1\},\{0\}, \{0\}, \{0\}\ldots, \{0\},\{0\}),$

which is a finite (better-reply)  improvement cycle.

  The following result provides existence of equilibria in behavioral  mixed strategies.

\begin{proposition} \label{pro:resubmission}
The multi-item Bayesian LUBA game (with resubmission) but with arbitrary number of bidders has  at least one Bayes-Nash equilibrium in mixed strategies under budget restrictions. Moreover, at a (generic) mixed equilibrium, the expected payoff of the bidder is zero.
\end{proposition}

Proposition \ref{pro:resubmission} provides existence of at least one Bayes-Nash equilibrium. However, it does not tell us what are those equilibria.

\begin{proof}[Proof of Proposition \ref{pro:resubmission}]
By Proposition \ref{domino} the constrained game has a finite number of action. By standard fixed-point theorem,
the multi-item Bayesian LUBA game (with  resubmission) but with arbitrary number of bidders has  at least one Bayes-Nash equilibrium in mixed strategies under budget restrictions.

\end{proof}

Below we explicitly compute mixed equilibrium for two bidders with resubmission.

\begin{proposition} \label{tworesubmission}Let $c +1< v$ and $n=2, c>0.$
The game has a partially mixed equilibrium which is explicitly given by
$$y^*=(\frac{c}{v-1},\frac{c}{v-2},\ldots, \frac{c}{v-k}, 1-\sum_{l=0}^{k-1} \frac{c}{v-(l+1)},0,\ldots, 0)$$ where $k$ is the maximum  number such that  $\sum_{l=0}^{\min(\bar{b},\frac{v}{c})}
\frac{\tilde{c}}{v-l}<1.$
\end{proposition}

 \begin{proof}[Proof of Proposition  \ref{tworesubmission}]
Let $k$ be the largest integer such that $y_k=\mathbb{P}(\{0,1, . . . , k \}) > 0,$ $k\leq \bar{b}$. When  bidder 2's strategy is $y,$ the expected payoff of  bidder 1 when bidding $\{0,1,...,l\}$ is equal to zero in equilibrium due to the indifference condition, for each $l \in \{1, 2,\ldots, k\}$. The cost of such a bid is equal to  $l c$.

On the other hand, the expected gain can be computed as follows. With probability $y_0$, bidder 2 will not post any bid, the winning bid is 1, and the gain is thus  $v - 1$ for any $l\leq  k.$
For $1\leq l\leq k,$ the bidder 2 bids $\{0,1,...,l\}$ with probability $y_l$ and the winner bid is $l+1$ from bidder 1, and 1's gain will $v-(l+1).$
The expected payoff of bidder 1 when playing $\{0,...,l\}$ is therefore given by
$$\begin{array}{l}
\mbox{Action} \{0\}:\\
r_{1i}(\{0\}, y)=0\\
\mbox{Action} \{01\}:\\
r_{1i}(\{01\}, y)= (v-c-1)y_0 -c y_1 -c(y_2+\ldots+y_k),\\
\mbox{Action} \{012\}:\\
r_{1i}(\{012\}, y)= (v-2c-1)y_0 +(v-2c-2)y_1\\ -2c y_2 -2c(y_3+\ldots+y_k),\\
\ldots \\ \mbox{Action} \{012\ldots l\}:\\
r_{1i}(\{012\ldots l\}, y)= (v-lc-1)y_0 \\ +(v-lc-2)y_1\\ +\ldots +(v-lc-l)y_{l-1}\\ -lc y_l -lc(y_{l+1}+\ldots+y_k)\\ \mbox{Action} \{012\ldots (l+1)\}:\\
r_{1i}(\{012\ldots l+1\}, y)= (v-(l+1)c-1)y_0\\ +(v-(l+1)c-2)y_1\\ +\ldots +(v-(l+1)c-(l+1))y_{l}\\ -(l+1)c y_{l+1} -(l+1)c(y_{l+2}+\ldots+y_k)\\
\ldots \\ \mbox{Action} \{012\ldots k\}:\\
r_{1i}(\{012\ldots k\}, y)=  (v-kc-1)y_0+(v-kc-2)y_1\\ +\ldots+(v-kc-k)y_{k-1}-kc y_k.\\
y_l\geq 0,\ y_0+\ldots+y_k=1\\
 y_{k+s+1}=0\ \mbox{for}\ s\geq 0.
\end{array}$$

It turns out that

$$\left\{\begin{array}{c}
 (v-1)y_0 = c\\
 (v-1)y_0 +(v-2)y_1=2c \\
\ldots \\
(v-1)y_0 +(v-2)y_1+\ldots +(v-l)y_{l-1}=lc \\
 (v-1)y_0 +(v-2)y_1+\ldots +(v-(l+1))y_{l}=(l+1)c \\
\ldots \\
  (v-kc-1)y_0+(v-kc-2)y_1+\ldots+(v-kc-k)y_{k-1}=kc.
\end{array} \right.
$$

For $l$ between $1$ and $k-1$ we make the difference between line $l+1$ and line $l$ to get:
$$\left\{
\begin{array}{c}
 y_0 = \frac{c}{v-1}\\
y_1=\frac{c}{v-2} \\
\ldots \\
y_{l-1}=\frac{c}{v-l} \\
y_{l}=\frac{c}{v-(l+1)} \\
\ldots \\
 y_{k-1}=\frac{c}{v-k} \\
 y_k=1-(y_0+y_1+\ldots+y_{k-1}) >0\\
 y_{k+1+s}=0.
\end{array}
\right.
$$

Thus, the partially mixed strategy $$y^*=(\frac{c}{v-1},\frac{c}{v-2},\ldots, \frac{c}{v-k}, 1-\sum_{l=0}^{k-1}\frac{c}{v-(l+1)},0,\ldots, 0)$$ is an equilibrium strategy. The equilibrium payoff is zero. This completes the proof.

\end{proof}

Note that the framework can easily capture situations in which the number of potential participants can be unbounded. We introduce the statistics of the bidding data as $n_{i,b}$ which is the number of bidders who place $b$
on item $i.$ The random matrix  $(n_{i,b})_{i,b}$  contains enough information that will allow any player to compute its payoff. Therefore one can work directly on bid statistics $(n_{i,b})_{i,b}.$ The knowledge of the total number of participants is not required.  The payoff function of $j$ on item $i$ is $[-c_r-c_i |B_{ji}|-(v_{ji}-b_{ji})\ind_{\{ b_{ji}=\inf B_i^*\}}]\ind_{\{  B_{ji}\neq \{0\}\}}.$
The dependence on the mean-field term $(n_{i,b})_{i,b}$ is expressed as following:
$b_{ji}\in B_i^*$ if and only if  $b_{ji}$ is the smallest non-zero bid such that $n_{i,b_{ji}}=1.$

\subsection{Revenue of the auctioneers}
We now investigate how much money the online platform can make by running multi-item LUBA. Since the platform will be running for a certain time before the auction ends, each bidder is facing a
a random number of other  bidders, who may bid in a stochastic strategic way. Their valuation is not known. We need to estimate the set of bids $B_i$ and the bid values on item $i.$
We denote by $n_{ib}$  the random number of bidders who bid on $b.$ If $n_{i,b}$ follows a Poisson distribution with parameter $\lambda_{i,b}$, and that all variables $n_{i,b}$ are independent. The value $\lambda_{i,b}$ is assumed to be non-decreasing with $b.$ The expected payoff of the seller on item $i$ is
$\sum_{j}c_r \ind_{B_{ji}\neq \{0\}}+ \sum_{b}\mathbb{E}n_{i,b} c_i +\mathbb{E} \inf B_i^*  - v_i.$
\begin{proposition}\label{resubmissionauctioneer}
Let $\lambda_{i,b}=\frac{v_i}{c_i} \frac{1}{(1+b)^{z}},$ with $z>0.$
The expected revenue of the seller on item $i$ is  $\sum_{j}c_r +\mathbb{E} \inf B_i^* + v_i [\sum_{b} \frac{1}{(1+b)^{z}}],$
exceeds the value $v_i$ of the item $i$  for small value of $z$ whenever  $$\sum_{b\leq \min(\bar{b}, \frac{v_i}{c_i})} \frac{1}{(1+b)^{z}} > 1 -\frac{\sum_{j}c_r +\mathbb{E} \inf B_i^*}{v_i}$$
\end{proposition}
\begin{proof}[Proof of Proposition \ref{resubmissionauctioneer}]

The expected payoff of the seller on item $i$ is equal to $\sum_{j}c_r + \sum_{b}\mathbb{E}n_{i,b} c_i +\mathbb{E} \inf B_i^*  - v_i.$ As we assume that $n_{i,b}$ follows a Poisson distribution with parameter $\lambda_{i,b}$ and that all $n_{i,b}$ are independent, we can calculate that $\mathbb{E}n_{i,b}=\lambda_{i,b}=\frac{v_i}{c_i} \frac{1}{(1+b)^{z}}.$ Rewriting  the expected payoff of the auctioneer, we  get that it is equal to $\sum_{j}c_r + \sum_{b}\mathbb{E}n_{i,b} c_i +\mathbb{E} \inf B_i^*  - v_i.$
Then, we can easily induce the condition for the expected revenue of the seller exceeds the value of $v_i$ is
 $$\sum_{b\leq \min(\bar{b}, \frac{v_i}{c_i})} \frac{1}{(1+b)^{z}} > 1 -\frac{\sum_{j}c_r +\mathbb{E} \inf B_i^*}{v_i}$$

\end{proof}

Note that in all-pay auction (with resubmission cost) if the probability to win is $P_j(B)=\frac{1}{n}$ (uniform lottery) then the revenue of the seller is  $\sum_{j}c_r \ind_{B_{ji}\neq \{0\}}+ \sum_{b}\mathbb{E}n_{i,b} c_i  - v_i.$ Thus, the difference between the two schemes is the expected winning bid $\mathbb{E} \inf B_i^*$ if the action profile generates similar distribution $n_{i,b}.$

Proposition \ref{pro:resubmission} provides existence of at least one Bayes-Nash equilibrium. However, it does not tell us how to learn or to reach  those equilibria.
Below we provide a learning procedure for equilibria.

\subsection{Learning Algorithm}

 \label{sec:5}

Imitative learning is very important in many applications
\cite{mastercourselearning} including cloud networking, power grid, security \& reliability, information dissemination and evolution of protocols and technologies.
It has been successfully used to capture  animal behavior as well as human learning.

Based on  randomly disturbed games \`a la Harsanyi \cite{harsanyi} who  showed that the equilibria of a game are limit points of sequences
of $\epsilon-$Nash  as the parameter $\epsilon$ vanishes, the imitative learning follows the same line as a perturbed payoff-based scheme. So, when trying to learn
the equilibria of a game it makes sense to consider imitative Boltzmann-Gibbs strategy with
a small  parameter $\epsilon.$  The number $\frac{1}{\epsilon}$ is sometimes interpreted as a rationality level of the player. Small $\epsilon$ is therefore seen as big rationality level.

 Surprisingly, it also makes sense to consider imitative Boltzmann-Gibbs with big parameter $\epsilon$ because
when the   parameter $\epsilon$ approaches infinity, the limiting of the imitative logit dynamics is the so-called
 ``imitate the better'' dynamics  \cite{weibull}.
 The actions that are initially not present in the support of the strategy will never be tried with imitative dynamics,
so we always consider initial points in the interior of the strategy space.

{\it What is the information-theoretic function associated to the imitative Boltzmann-Gibbs (iBG) strategy}?
The information-theoretic metric behind the imitative Boltzmann-Gibbs learning is the relative entropy. This is an important connection between imitative  learning and information theory since the relative entropy covers the mutual information, information gain and the Shannon entropy as particular case. We introduce relative entropy as a cost of moves in the LUBA games. One can interpret relative entropy as a cost of moves in the  learning process.
 Specially, the next Proposition \ref{propoimitative} shows that
 the imitative Boltzmann-Gibbs strategy (or imitative logit strategy, i-logit) is the maximizer of the perturbed payoff
$r_j(s'_j,s_{-j})-\epsilon_j d_{KL}(s'_j,s_j)$ where $d_{KL}$ is the relative entropy from strategy $s_{j}$ to  $s'_{j}$
 (also called Kullback-Leibler divergence)  and $\frac{1}{\epsilon_j}$ is the rationality level of player $j,$
$d_{KL}(s'_j,s_j)=-\sum_{i=1}^{+\infty}s'_{ji}\log_2 \left(\frac{s'_{ji}}{s_{ji}}\right).$

The instant payoff of bidder $j$ on item $i$ at time/round $t$ is a realized value of
$$R_{ji}^t=r_{ji}( B^t)+\eta_{ji}^t $$ where $\eta_{ji}^t$ is an  observation/measurement noise.
Bidder $j$ wins  on item $i$ at time $t$ if its bid  $b_{ji,t}$ is the lowest unique bid on item $i$ and $B_{ji,t}$ is feasible in terms of the available budget of $j$ at time $t.$ $R_{ji,t}=0$ if bidder $j$ does not participate to item $i.$
This algorithm describes how to update the reward and the strategy in the bid space. Let $\hat{R}_{ji,t}(B_{j})$ be the estimation of reward corresponding the $j$-th bidder to $i$-th item at round $t$ if she decides the bid set corresponding to the index of $B_j,$ and $\hat{R}_{j,t}(B_{j})=\sum_{i}\hat{R}_{ji,t}(B_{j})$
\begin{algorithm}
\caption {The proposed update reward learning algorithm}
\label{algorithm: 001}
\textbf{Initialization:} Estimate reward on item $i$ at round zero.
$\hat{R}_{ji,0}(0), \hat{R}_{ji,0}(1), \ldots,\hat{R}_{ji,0}(\bar{b}(0))\sim uniform$

\textbf{For Round $t+1$} \\
For every bidder and every item:\\

$\hat{R}_{j,t+1}(B_j)=\hat{R}_{j,t}(B_j)+\ind_{\{B_{j,t}=B_j\}}\alpha_{j}^t(R_{j,t}-\hat{R}_{j,t}({B_j}))$\\

\textbf{Update  S:} $s_{j,t+1}(B_j)=s_{j,t}(B_j)(1+\lambda_{j,t})^{\hat{R}_{j,t}(B_j)}$ \\
\textbf{Normalize S:}
\end{algorithm}

 Let $s_{j,t}$ be a vector in the relative interior of the $(|\mathcal{B}_j|-1)$-dimensional simplex of $\mathbb{R}^{|\mathcal{B}_j|}.$
 This means that the cost to move for player $j,$ from $s_{j,t}$ to $s_{j,t+1}$ is given by
$$c(s_{j,t+1},s_{j,t})=\epsilon_{j,t} \sum_{B_j\in \mbox{supp}(s_{j,t+1})}s_{j,t+1}(B_j) \ln\left(\frac{s_{j,t+1}(B_j)}{s_{j,t}(B_j)}\right),$$ where $\epsilon_{j,t}>0.$ This cost is added to the expected estimated payoff which is $\langle s_{j,t+1}, \hat{r}_{j,t}\rangle.$ Thus, we associate the following problem to player $j:$
$$\max_{s_{j,t+1}}\left[\langle s_{j,t+1}, \hat{r}_{j,t}\rangle- c(s_{j,t+1},s_{j,t}) \right].$$

\begin{eqnarray}
\hat{R}_{j,t}&=&  \max_{B_j}\hat{r}_{j,t}(B_j)
= \max_{s_{j,t+1}\in \mathbb{P}(\mathcal{B}_j)} \langle \hat{r}_{j,t}, s_{j,t+1}\rangle
\end{eqnarray}
We introduce the Legendre-Fenchel transform of the relative entropy function
\begin{eqnarray}\label{eqtsrta}
W_{j,\epsilon_j}&=& \max_{s_{j,t+1}\in  \mathbb{P}(\mathcal{B}_j)} \tilde{W}_{j,\epsilon_j}(s_{j,t+1}| s_{j,t},\hat{r}_{j,t})
\end{eqnarray} where
$\tilde{W}_{j,\epsilon_j}(s_{j,t+1}| s_{j,t},\hat{r}_{j,t}):=\left[\langle \hat{r}_{j,t}, s_{j,t+1}\rangle-\epsilon_{j,t} \sum_{B_j\in \mbox{supp}(s_{j,t+1})}s_{j,t+1}(B_j) \ln(\frac{s_{j,t+1}(B_j)}{s_{j,t}(B_j)})\right],$
$\epsilon_{j,t}$ is a positive parameter. $W_{j,\epsilon_j}$ corresponds to the best expected estimated  payoff with cost of learning.
\begin{proposition} \label{propoimitative}
The following statements hold: \\

(i) The strategy
$
s_{j,t+1}(B_j)=\frac{s_{j,t}(B_j)e^{\frac{ \hat{r}_{j,t}(B_j)}{\epsilon_{j,t}}}}{\sum_{B'_j} s_{j,t}(B'_j)e^{\frac{\hat{r}_{j,t}(B'_j) }{\epsilon_{j,t}}}}
$
is the optimal strategy solution to (\ref{eqtsrta}).

(ii)
The Lagrange multiplier associated to (\ref{eqtsrta}) is given by
$$
\nu_j=\epsilon_{j,t}\left[ -1+ \ln\left(\sum_{B_j} s_{j,t}(B_j)e^{\frac{\hat{r}_{j,t}(B_j) }{\epsilon_{j,t}}}\right)\right],
$$
and the optimal value is
$$
W_{j,\epsilon_j}=\epsilon_{j,t} \ln\left(\sum_{B_j} s_{j,t}(B_j)e^{\frac{\hat{r}_{j,t}(B_j) }{\epsilon_{j,t}}}\right),\ \epsilon_{j,t}>0.
$$
\end{proposition}

 \begin{proof}[Proof of Proposition \ref{propoimitative}]
 Let $\epsilon_{j,t}>0.$
The  function
$$s_{j,t+1}\longmapsto -\epsilon_{j,t}\sum_{B_j\in \ \mbox{support}(s_j)}s_{j,t+1}(B_j) \log_2(s_{j,t+1}(B_j))-constant$$  is strictly concave (the Hessian matrix
with entries $h_{ii}<0$ and $h_{ij}=0$) and the domain (simplex) is convex. Thus, the Karush-Kuhn-Tucker (KKT) conditions are necessarily and sufficient for the problem in (\ref{eqtsrta}). Using KKT conditions, one has
$$ \hat{r}_{j,t}(B_j)+\epsilon_{j,t}\ln  s_{j,t}(B_j)-\epsilon_{j,t} (1+\ln s_{j,t+1}(B_j))-\nu_j=0,$$  where $\nu_j$ is the Lagrange multiplier associated to the simplex equality equation of player $j$.
It follows that
$$
\frac{\hat{r}_{j,t}(B_j)-\nu_j}{\epsilon_{j,t}}-1=\ln\left( \frac{s_{j,t+1}(B_j)}{s_{j,t}(B_j)} \right)
$$

Taking the exponential yields
$$
s_{j,t+1}(B_j)=\frac{s_{j,t}(B_j)e^{\frac{\hat{r}_{j,t}(B_j)}{\epsilon_{j,t}}}}{e^{1+
\frac{\nu_j}{\epsilon_{j,t}}}}
$$
Summing over the set of actions, one gets
$\sum_{B'_j}s_{j,t+1}(B'_j)=1$ which implies that
$$
e^{1+\frac{\nu_j}{\epsilon_{j,t}}}=\sum_{B'_j} s_{j,t+1}(B'_j)e^{\frac{\hat{r}_{j,t}(B'_j)}{\epsilon_{j,t}}}
$$

Thus, the optimal strategy is
$$
s_{j,t+1}(B_j)=\frac{s_{j,t}(B_j)e^{\frac{ \hat{r}_{j,t}(B_j)}{\epsilon_j}}}{\sum_{B'_j} s_{j,t}(B'_j)e^{\frac{\hat{r}_{j,t}(B'_j) }{\epsilon_j}}}.
$$

The Lagrange multiplier is
$$
\nu_j=\epsilon_{j,t}\left[ -1+\ln\left(\sum_{B_j} s_{j,t}(B_j)e^{\frac{\hat{r}_{j,t}(B_j)}{\epsilon_{j,t}}}\right)\right],
$$
and
$$
W_{j,\epsilon_j}=\nu_j+\epsilon_{j,t}=\epsilon_{j,t}\ln\left(\sum_{B_j} s_{j,t}(B_j)e^{\frac{\hat{r}_{j,t}(B_j)}{\epsilon_{j,t}}}\right)
$$
\end{proof}

The next Proposition specifies the error bound in terms of the parameter $\epsilon_{j,t}.$

\begin{proposition} \label{progap}
There exists $C_{|\mathcal{B}_j|}>0$ such that
\begin{eqnarray}
|W_{j,\epsilon_j}-\max_{B_j} \hat{r}_{j,t}(B_j)|=|W_{j,\epsilon_j}-\hat{R}_{j,t}| \leq C_{|\mathcal{B}_j|}\epsilon_{j,t}
\end{eqnarray}
Moreover,  $C_{|\mathcal{B}_j|}\leq \ln\left(|\mathcal{B}_j|\right)+\max_{B_j}|-\ln(s_{j,t}(B_j))|.$
\end{proposition}
\begin{proof}[Proof of Proposition \ref{progap}]
The proof follows immediately from the definition of $W_{j,\epsilon_j}$ and the fact that the entropy function of player $j$ is between $0$ and $\ln |\mathcal{B}_j|.$
\end{proof}

\begin{proposition} \label{progap2} The following results hold:
\begin{itemize}\item {\bf High rationality regime: }
 $$\lim_{\epsilon \longrightarrow 0}W_{j,\epsilon}=\max_{B_j}\ \hat{r}_j(B_j)$$
\item {\bf Low rationality regime: }
 $$\lim_{\epsilon \longrightarrow \infty}W_{j,\epsilon}= \langle s_{j,t}, \hat{r}_{j,t} \rangle$$
\end{itemize}
\end{proposition}
\begin{proof}[Proof of Proposition \ref{progap2}]

The first limit follows from Proposition \ref{progap}. We now prove the second limit.
By changing the variable $\epsilon'=\frac{1}{\epsilon}$ one gets that $\epsilon'$ goes to zero and the limit becomes
\begin{eqnarray}
\lim_{\epsilon'\longrightarrow 0} \
\frac{\ln\left(\sum_{B_j} s_{j}(B_j)e^{\epsilon'\hat{r}_{j}(B_j)}\right)}{\epsilon'}\\ = \nonumber
\lim_{\epsilon'\longrightarrow 0} \
\frac{\left(\sum_{B_j} s_{j}(B_j)\hat{r}_{j}(B_j)e^{\epsilon'\hat{r}_{j}(B_j)}\right)}{\sum_{B_j} s_{j}(B_j)e^{\epsilon'\hat{r}_{j}(B_j)}}\\
=\sum_{B_j} s_{j}(B_j)\hat{r}_{j}(B_j)
\end{eqnarray}
by Hospital's rule. This completes the proof.
\end{proof}

It is important to notice that when $\epsilon_j$ goes to zero then the rationality level   $\frac{1}{\epsilon_j}$ tends infinity and player $j$ gets closer to the maximum payoff $\max_{B_j}\ \hat{r}_j(B_j).$
This result also says that when $\epsilon$ goes to infinity, one gets the expected payoff. Hence it gives a stationary point of  the replicator equation. We retrieve a well-known result in evolutionary process which states   that the evolution of phenotypes and genes has a tendency  NOT to  maximize the fitness but the {\it mixability} of system which is somewhat captured with the mixture  term $ \langle s_{j,t}, \hat{r}_{j,t} \rangle.$

\subsubsection*{Relationship with replicator dynamics}
Following \cite{anor} the scaled stochastic process from $s$ follows a replicator dynamics given by

$$
\dot{s}_{j}(B_j)=s_{j}(B_j)[ \mathbb{E}\hat{R}_{j}(B_j) -\sum_{B'_j} s_{j}(B'_j)\mathbb{E}\hat{R}_{j}(B'_j) ],\
$$
when the learning rate  $ (\lambda_{ji},\alpha_{ji}) ,$ which is random matrix,   vanishes.



\subsection{Second Learning Algorithm}
The complexity of the bid space in the lowest unique bid system is $O(2^N)$. Take $5\$$ for an example, the bid space of LUBA with resubmission is $2^500$ when the unit is 1 cent. This limits considerably the usability of the proposed algorithm. There is a need for reducing the curse of dimensionality or complexity. We utilize Monte Carlo to approximate the Nash-Equilibrium. The basic idea is that we utilize information provided by the system to guide the behavior of bidders, and then the frequency of winner's output is utilized as the approximation of Nash-Equilibrium.

During the bidding process, after placing a bid, the following information is provided by the system to the bidder.
\begin{enumerate}
\item Currently, whether the bid $k$ wins or not.
\item If not, the reason is:
\begin{enumerate}
\item $k$ is non-unique.
\item $k$ is too high.
\end{enumerate}

\end{enumerate}

\begin{algorithm}
\caption {The proposed Monte-Carlo algorithms}
\label{algorithm: 001}
\textbf{Initialization:}
\begin{enumerate}
\item For bidder $j$, initial $\delta_k^0=\frac{1}{\overline{b}_j}$
\item Generate bid $k$ from $\delta_k^{t-1}$
\item Based on the system output information, update $\delta_k^t$
\begin{enumerate}
\item $k$ is non-unique: $\delta_k^t=0$. $\delta_{l\not\in|K|}^t=\delta_{l\not\in|K|}^{t-1}+\frac{\delta_{k}^{t-1}}{\overline{b}_j-|K|+1}$
\item $k$ is too-high: $\delta_k^t=0$
    $\delta_{l\not\in|K|l\leq k}^t=$$\delta_{l\not\in|K|l\leq k}^t-1+\frac{\delta_{k}^{t-1}}{k-|K|+1}$
\end{enumerate}
\end{enumerate}
\textbf{For Round $t+1$} \\
For every bidder and every item:\\

$\hat{R}_{j,t+1}(B_j)=\hat{R}_{j,t}(B_j)+\ind_{\{B_{j,t}=B_j\}}\alpha_{j}^t(R_{j,t}-\hat{R}_{j,t}({B_j}))$\\

\textbf{Update  S:} $s_{j,t+1}(B_j)=s_{j,t}(B_j)(1+\lambda_{j,t})^{\hat{R}_{j,t}(B_j)}$ \\
\textbf{Normalize S:}
\end{algorithm}

\subsection{Pareto optimality and global optima for bidders}
Pareto optimality (PO)  for bidders', is a configuration in which  it is not possible to make any one bidder better off without making at least one bidder worse off.
\begin{proposition}
Let $ v_{ji} > c+c_r+1.$
\begin{eqnarray}
S_2&=& \{ (B_{ji})_{j,i}\ | \  \forall i,\  n_{i,1}=1,\ n_{i,b}=0, \ \forall b\geq 2 \nonumber \\
&&  (B_{ji})_i\in \mathcal{B}_j \}.
\end{eqnarray}
Then, any profile in  $S_2$  is a Pareto optimal solution and the global payoff of the bidders at these PO is
$$
 \sum_{i=1}^m(-c_r - c+ v_{ji}-1)\ind_{j | b_{ji}=1=\inf B_i^*}.
$$

The global optimum (GO) payoff of the bidders consists to select the maximum value $\max_{j}v_{ji}$ for each item and place 1 cent for that bidder if its budget allows to do so.

\end{proposition}

The proof is immediate.

The inefficiency gap between total bidders' payoff at mixed Bayes-Nash equilibrium (which is $0$)  and GO payoff is
$$
 \sum_{i=1}^m(-c_r -c+ \max_{j}v_{ji}-1)\ind_{b_{ji}=1=\inf B_i^*}.
$$

\subsection{Numerical Investigation}
In this subsection we introduce the experiment setting for the learning equilibria of LUBA  game.

\subsubsection{Illustration of Proposition \ref{tworesubmission}: 2 bidders for 1 item }
First we introduce the simulation of Proposition \ref{tworesubmission}. Table  \ref{tableeset} shows the experiment setting. We observe from Figure \ref{Nashequ1} that the algorithm provides a satisfactory result for approaching the mixed Bayes-Nash equilibrium distribution.
\begin{table}[ht]
\caption{Two bidders - One item under budget restriction}
\begin{center}
\begin{tabular}{|c|c|c|c|} \hline
     $n$       & 2          &  $m $       &  1   \\ \hline
$\overline{b}$ & 6          & $ v  $      &  8    \\ \hline
  $\alpha$     & 0.5        & $\lambda$   &  0.1   \\ \hline
  Assumption   & Symmetric  & Assumption  & Static Budget\\ \hline
\end{tabular}
\end{center}
\label{tableeset}
\end{table}%

\begin{figure}[!htb]
  \centering
    \includegraphics[width=10cm]{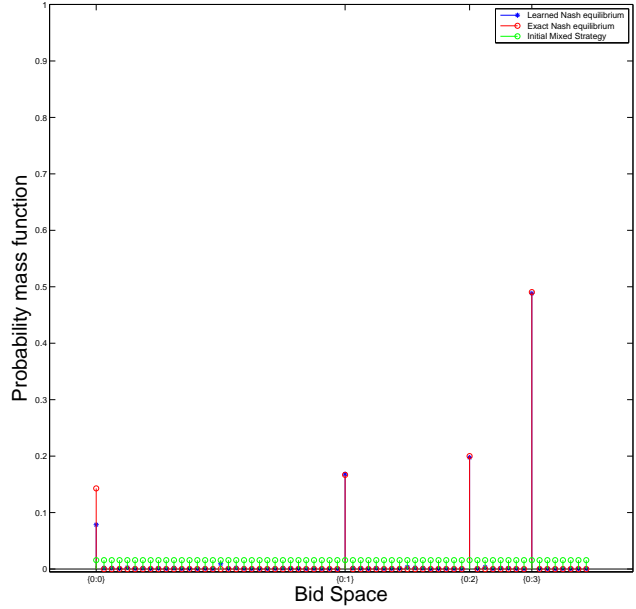}
  \caption{Probability mass function of 2 bidders and 1 Item in the auction system with resubmission. The blue asterisks signs present learned probability mass functions of bidders by the proposed algorithm in final iteration, the red circle signs present the theoretical Nash equilibrium provided by Proposition \ref{tworesubmission},  the green circle signs present the initial mixed strategy assigned to bidders in the beginning of learning process.} \label{Nashequ1}
\end{figure}
\subsubsection{Two players for one item (asymmetric situation)}
Let $n=2, m=1, \bar{b}_2>\bar{b}_1, v_2>v_1$ The action profile of bidder 1 is $\{\{0\},\{0,1\},\{0,1,2\}\}$ and that of bidder 2 is $\{\{0\},\{0,1\},\{0,1,2\},\{0,1,2,3\}\}.$ Registration fee $c_r=1$ and submission cost $\ c=1$ The parameters assumption in this simulation conveys different bidders have different valuation target same item. Table \ref{tabasy1} represents the payoffs matrix of bidder 1 and 2. The row player is player 1 and the column player is player 2. The ceil list rewards of each bidder and bidder 1 come first.
\begin{table}[ht]
\tiny{
\caption{Two bidders - One item under budget restriction under asymmetric situation}
\begin{center}
\begin{tabular}{|c|c|c|c|c|} \hline
            & \{0\}       & \{01\}       & \{012\}     & \{0123\}     \\ \hline
  \{0\}     & $0,0$     & $0,v-c-1$   & $0,v-2c-1$  & $0,v-3c-1$    \\ \hline
  \{01\}   & $v-c-1, 0$ & $-c,-c$     & $-c,v-2c-2$ & $-c,v-3c-2$    \\ \hline
  \{012\} & $v-2c-1, 0$& $v-2c-1,-c$ & $-2c,-2c$   & $-2c,v-3c-3$    \\ \hline
\end{tabular}
\end{center}
\label{tabasy1}
}
\end{table}%
Let the mixed strategy of bidder 1  be $s_1=x={x_0, x_1, x_2}$ and that of bidder 2 $s_2=y={y_0,y_1,y_2,y_3}$ . According to the indifferent condition we can derive the mixed Nash equilibrium is $x=(1-\frac{c}{v-2}-\frac{c}{v-3},\frac{c}{v-2},\frac{c}{v-3}); y=(\frac{c}{v-1},\frac{c}{v-2},\frac{c}{v-3},1-\frac{c}{v-1}-\frac{c}{v-2}-\frac{c}{v-3})$.
We now  analyze the mixed strategy of bidder 1. According to the indifferent condition, we can derive the following set of equations.
$$\begin{array}{l}
\mbox{Action} \{0\}:\\
r_{1}(\{0\}, x)=0\\
\mbox{Action} \{01\}:\\
r_{1}(\{01\}, x)= (v-c-1)x_0 -c(x_1+x_2),\\
\mbox{Action} \{012\}:\\
r_{1}(\{012\}, x)= (v-2c-1)x_0 +(v-2c-2)x_1-2cx_2\\
\mbox{Action} \{0123\}:\\
r_{1}(\{0123\}, y)= (v-3c-1)x_0 +(v-3c-2)x_1 \\+(v-3c-3)x_3
\end{array}$$
As the action $\{0123\}$   provides an expected payoff which is strictly higher than $0$ obtained with the action $\{0\}.$ The action $\{0\}$ is not  in the  support of the mixed Nash equilibrium. Assume the rewards of action $\{0,1\},\{0,1,2\}$ and $\{0,1,2,3\}$ are equal, then we can derive the Nash equilibrium of bidder 1 is $x=(1-\frac{c}{v-2}-\frac{c}{v-3},\frac{c}{v-2},\frac{c}{v-3}); $ The Nash equilibrium of bidder 2 can be derived utilizing the same method.
\subsubsection{Three bidders - One item}
Let $n=3, m=1, \bar{b}_1=5, \ \bar{b}_2=\bar{b}_3=3, \ c=c_r=1.$ Table \ref{taby1} represents the payoffs of bidder 1. The bid profile of  bidders 2 and 3 are displayed in the column. Bidder 1 choice is a row of the matrix.
%
%

\begin{table}[ht]
\caption{Three bidders - One item under budget restriction $\bar{b}_1=5, \ \bar{b}_2=\bar{b}_3=3, $ registration fee $c_r=1$ and submission cost $\ c=1.$}
\begin{center}
\begin{tabular}{|c|c|c|c|c|} \hline
 & \{0\} \{0\} & \{1\} \{0\}  & \{0\} \{1\}  & \{1\} \{1\} \\ \hline
  \{0\}& 0& 0 &  0& 0\\ \hline
    \{1\}& $(v-3)^*$&  $-2$&  $-2$& $-2$\\ \hline
      \{2\}& $v-4$&   $-2$&   $-2$& $v-4$\\ \hline
       \{ 3\}& $v-5 $&  $-2$&  $-2$& $v-5$\\ \hline
          \{12\}& $v-4$ & $v-5$ &  $v-5$&$v-5$ \\ \hline
            \{13\}& $v-4 $&  $-3$&  $-3$& $-3$ \\ \hline
              \{23\}& $v-5$& $-3$ & $-3$ & $v-5$\\  \hline
                \{123\}& $v-5$&  $-4$& $-4$ &$-4$ \\ \hline

\end{tabular}
\end{center}
\label{taby1}
\end{table}%

If $v-c-c_r-1=v-3<0$ then player 1 will not participate because the value of the item is below the cost to get it.
Due to the unfeasible budget from bidders 2 and 3 who cannot bid 2 cents, the pure action $(\{1\}, \{0\},\{0\})$ is a Bayes-Nash equilibrium if $v-c-c_r-1=v-3\geq 0.$
If $v-2c-c_r-2=v-5>0$ then bidder 1 can guarantee a positive payoff by playing the action $\{12\}.$ Therefore, player $1$ will participate for sure because her action $\{12\}$ dominates $\{0\}$ in that case.

From Figure \ref{PureNash}, which is the corresponding simulation results, we can conclude that the learning algorithms converged to the Nash equilibrium described in Table \ref{taby1}.
\begin{figure}[!htb]
  \centering
    \includegraphics[width=0.9\textwidth]{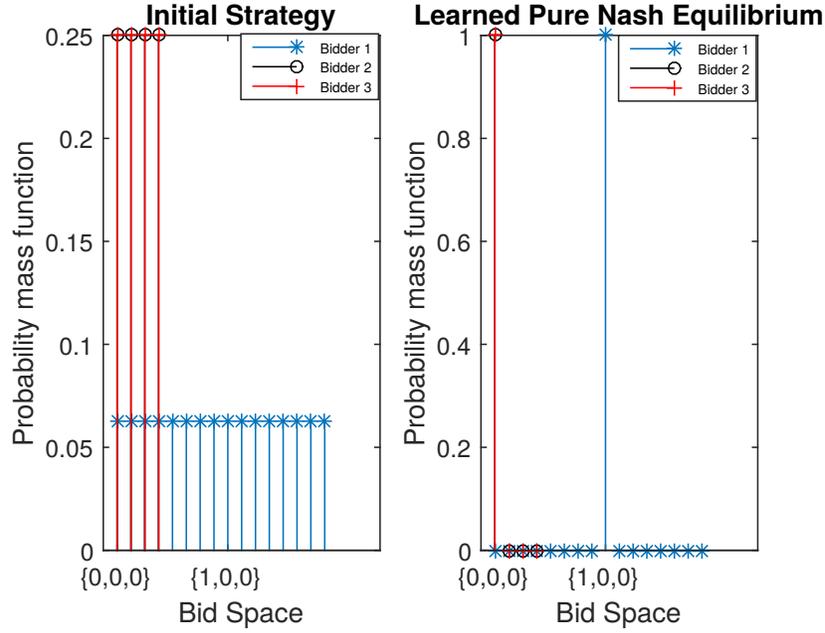}
  \caption{Probability mass function of 3 players and 1 Items in the auction system with resubmission and budget constraint. The blue asterisks present Bidder 1's strategy. The black circles present Bidder 2's strategy and the read plus sign present Bidder 3's strategy.} \label{PureNash}
\end{figure}

\subsubsection{Four bidders - Two items}

\begin{table}[!htb]
\caption{Summary of Multi-Items Experiment Setting } \label{tableAss11}
\begin{center}
\begin{tabular}{ll}
\hline
   Symbol & Setting\\ \hline
   $n$ & 4  \\
   $m$ & 2\\
   $\alpha$ & 0.001\\
   $\lambda$&0.01\\
   $c$ & 1\\
   Resource for item 1 and 2 & [2000 1500]\\
   Initial of Budget & [100 120 80 90]\\
   Initial of $\hat{R}_{ji}^{0}$ & 0.0001\\
   $v$ for $ j\in \{1,2,3,4\}$&[40,102],[42,109],[38,100],[36,110]\\
  \hline

\end{tabular}
\end{center}
\end{table}
\begin{figure}[!htb]
  \centering
    \includegraphics[width=0.9\textwidth]{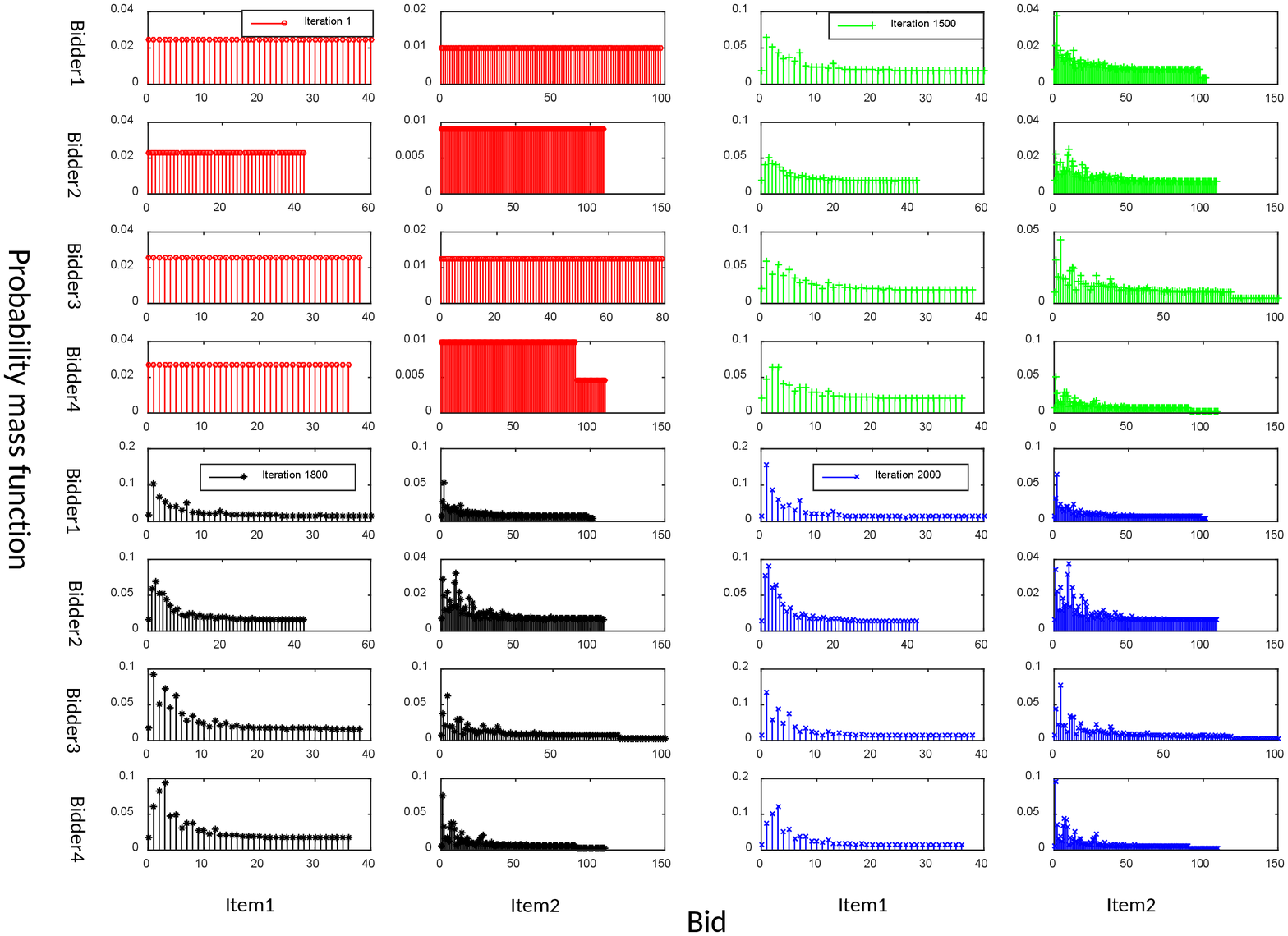}
  \caption{Probability mass function of 4 players and 2 Items in the auction system with budget update. The red circles present learned probability mass functions of each bidder target each item by the proposed algorithm in first iteration, the green plus sign present corresponding results in 1000th iteration, the black asterisks present results in 1800th iteration and the blue crosses present results in 2000th iteration.} \label{Nashequ3}
\end{figure}

\subsubsection{Impact of Parameters}
We investigate the impact of parameters in the proposed learning algorithm described.Table \ref{TablePara} shows experiment setting details.   In order to analyze the numerical convergence of the learning scheme we introduce root mean square error ˆRMSE ‰between two sequential round strategies.
$$RMSE_{t,t-1}=\sqrt{\sum_{j\in \mathcal{J}}\sum_{B_j}\sum_{i}(s_{ji,t}({B_j})-s_{ji,t-1}(B_j))^2}$$
Figure \ref{Paralem} and \ref{Paraalf} present the statistic properties of RMSE evolution with the bid round obtained by the proposed learning algorithm. Compared to $\alpha=0.1$ and $\lambda=0.1$, the results shows a quickly converged property. The plot shows that the large parameter setting results less outliers in the experiment results. According to the results in Figure \ref{Paralem} and \ref{Paraalf}, the parameter $\alpha$ and $\lambda$ influence the convergence of the proposed algorithm equally and a large $\alpha$ can reduce disturbance and outliers more effectively.

\begin{table}[htb]
\caption{ Parameters Impact Investigation} \label{TablePara}
\begin{center}
\begin{tabular}{lll}
\hline
   Symbol & Original Setting & Compared Setting\\ \hline
   $\alpha$ & 0.1 &1\\
   $\lambda$&0.1 &1\\
   \hline
\end{tabular}
\end{center}
\end{table}

\begin{figure}[!htb]
  \centering
    \includegraphics[width=0.9\textwidth]{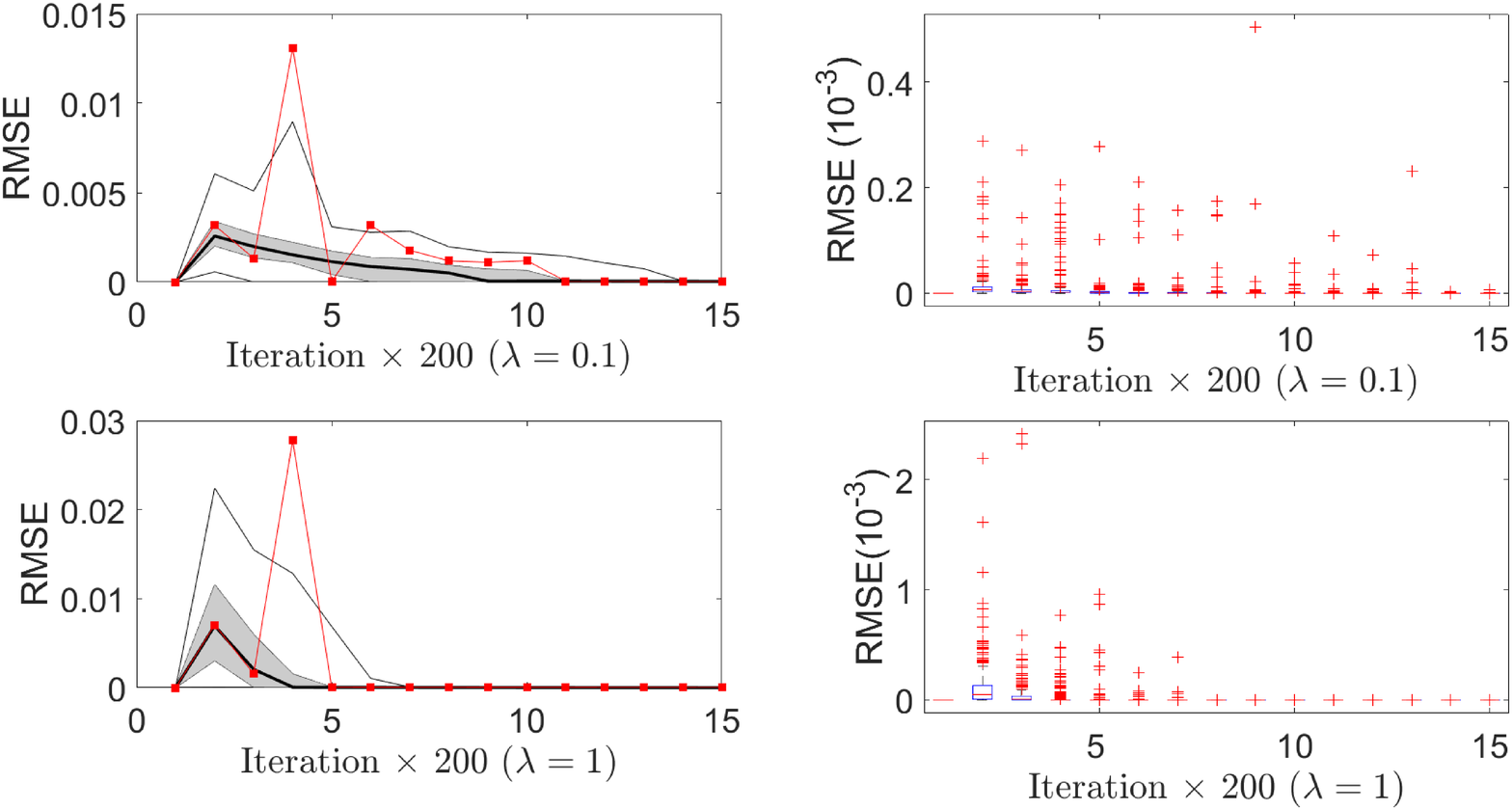}
     \caption{Statistic information on RSME on strategies probability distribution according to the experiment setting in the Table\ref{TablePara} . In the Figure \textbf{(a)} and \textbf{(c)}, the red curve with solid squares is the RMSE of 110th repeat experiment, the solid black curve is the median, the gray-shade area corresponding to the region between the percentiles P25 and P75, and the external bounding curves are the percentiles P5 and P95.\textbf{(b)} and \textbf{(d)} Box plot of RMSE on strategies probability distribution. On each box, the red central mark is the median, the edges of each box are the 25th and 75th percentage, the whiskers extend to the most extreme datapoints which are not considered to be outliers, and the outliers are plotted individually in the figure.}
  \label{Paralem}
\end{figure}
\begin{figure}[!htb]
  \centering
    \includegraphics[width=0.9\textwidth]{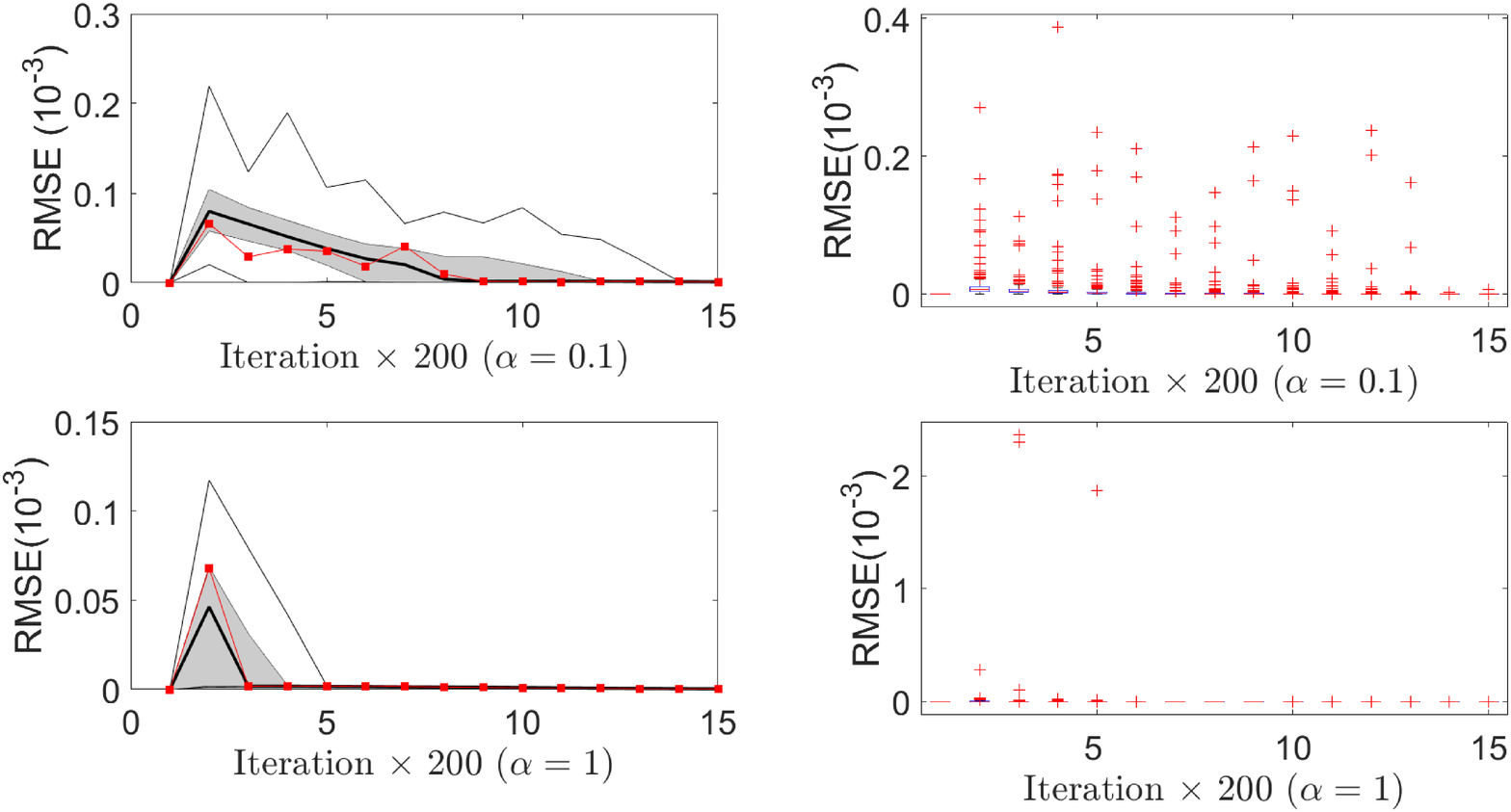}
     \caption{Statistic information on RSME on strategies probability distribution according to the experiment setting in the Table\ref{TablePara} . In the Figure \textbf{(a)} and \textbf{(c)}, the red curve with solid squares is the RMSE of 110th repeat experiment, the solid black curve is the median, the gray-shade area corresponding to the region between the percentiles P25 and P75, and the external bounding curves are the percentiles P5 and P95.\textbf{(b)} and \textbf{(d)} Box plot of RMSE on strategies probability distribution. On each box, the red central mark is the median, the edges of each box are the 25th and 75th percentage, the whiskers extend to the most extreme datapoints which are not considered to be outliers, and the outliers are plotted individually in the figure.}
  \label{Paraalf}
\end{figure}
\section{Risk-Sensitive Problem Statement}
\label{sec:2risk}

\subsection*{Multi-Item contest} A multi-item contest is a situation in which players exert effort for each item in an attempt to win a prize.
Decision to participate or not is costly and left to the players. An important ingredient in describing a multi-item contest  with the set of potential participants $\mathcal{J}=
\{1, \ldots , n\},$ the set of auctioneers proposing the set of items $\mathcal{I}=\{1,\ldots, m\}$
 is the contest success function, which takes the efforts $b$  of the agents and
converts them into each agent's probability of winning per item: $p_j : \ {b}=(b_{ji})_{j,i}\in \mathbb{R}^{n\times m} \mapsto [0, 1]^m. $ The (expected)
payoff of a risk-neutral player $j$ with  $v_{ji}$ the realized value of winning  item $i$ and a cost of effort function $c_{ji}$ is
 $\sum_{i\in \mathcal{I}}p_{ji}(b) v_{ji} -c_{ji}(b).$ The risk sensitive payoff is   $ u^{-1}[\mathbb{E}(u(v_{ji} -c_{ji}))],$ where $u$ is a  one-to-one mapping.

\subsection*{Multi-Item LUBA: Risk-Sensitive Case}
Multiple sellers (auctioneers) have multiple items (objects)   to sell on the online platform.  
 The payoff of bidder $j$ on item $i$ at that round would be
 $
 r_{ji}=v_{ji}- c_i - \inf B_{i}^*-c_r,
 $ if $j$ is a winner on item $i$, and
 $
 r_{ji}=-c_i - c_r,
 $ if $j$ is not a winner on item $i$ and $b_{ji}>0.$   The instant payoff of bidder $j$ on item $i$ is zero if $b_{ji}$ is reduced to $0$ (or equivalently the empty set).

 \begin{eqnarray}&& r_{ji} (b) \\  \nonumber
 &=& [-c_r-c_i +(v_{ji}-b_{ji})\ind_{\{ b_{ji}=\inf B_i^*\}}]\ind_{\{  b_{ji}\neq 0\}},
 \end{eqnarray}
where the infininum of the empty set is zero.

  \begin{eqnarray} r_j (b) =\sum_{i\in \mathcal{I}} r_{ji}(b).
 \end{eqnarray}

  The {\it risk-sensitive} instant payoff of  bidder $j$ is $ u^{-1}_j\left[\mathbb{E} u_j(r_{j})\right].$
Each bidder $j$ has a budget constraint
 $$
 c_r+\sum_{i\in\mathcal{I}} c_i\ind_{\{ b_{ji}\neq 0\}} +b_{ji}\ind_{\{ b_{ji}=\inf B_i^*\}}\leq \bar{b}_j.
 $$

 The instant payoff of the auctioneer of item $i$ is $$ r_{a,i}=\left(\sum_{j}(c_r+c_i)\ind_{\{ b_{ji}\neq 0 \}}+ \inf B_{i}^* \right)  -  v_{a,i},$$ where
  $v_{a,i}$ is the realized valuation of the auctioneer for item $i.$
 The instant payoff of the auctioneer of a set of item $I$ is $ r_{a,\mathcal{I}}=\sum_{i\in \mathcal{I}} r_{a,i}.$
 The {\it risk-sensitive} instant payoff of the auctioneer of a set of item $I$ is $ u_a^{-1}\left[\mathbb{E} u_a(r_{a,\mathcal{I}})\right].$

Bidders are interested in optimizing their risk-aware payoffs and the auctioneers are interested in their revenue under risk.

\subsection{Risk-Sensitive Solution Concepts}
Since the risk-sensitive game is of incomplete information, the strategies must be specified as a function of the information structure.
 \begin{defi}
 A pure strategy of  bidder  $j$ is a choice of a subset of natural numbers given the own-value and own-budget.
Thus, given its own valuation vector
 $v_j=(v_{ji})_i,$ bidder $j$ will choose an action  $(b_{ji})_i$  that satisfies the budget constraints $$\  c_r\ind_{\{ b_{j}\neq 0 \}}+\sum_{i=1}^m  [\inf B_{i}^*]\ind_{b_{ji} = \inf B_{i}^*}+\sum_{i=1}^m c_i\ind_{\{ b_{ji}\neq 0 \}}\leq \bar{b}_j.$$  The set of multi-item bid space for bidder $j$ is
 $$\begin{array}{c}
 \mathcal{B}_j(v_j,\bar{b}_j)=\{ (b_{ji})_i \ | \  \  b_{ji} \in [0, \bar{b}_j],\  \\   c_r\ind_{\{ b_{j}\neq 0 \}}+\sum_{i=1}^m  [\inf B_{i}^*]\ind_{b_{ji} = \inf B_{i}^*}+\sum_{i=1}^m c_i\ind_{\{ b_{ji}\neq 0 \}}\leq \bar{b}_j
\}.
 \end{array}$$

 A pure strategy is a mapping $ v_j \mapsto b_j\in \mathbb{R}_+.$
  A  constrained pure strategy is a mapping $ v_j \mapsto b_j\in \mathcal{B}_j.$ A mixed strategy is a probability measure over the set of pure strategies.

  \end{defi}

  \subsection{Continuum approximation of discrete bid space}
  Let $\delta>0$ be a currency/coin point in which the bidding will take place. $\delta$ is a rational number. The bid space is $\delta \mathbb{N}=\{ \delta, 2\delta,3\delta,\ldots\} \subset \mathbb{R}.$
When the bid units are very small (below or in the order of  $\delta=\frac{1}{100}$) the bid space is huge and the complexity is exponential.  One standard method  is to consider this set as discretization of the  continuous set $\prod_{i\in \mathcal{I}}(0, \min(\bar{b}_j,{v}_{ji}-c_i-c_r)]$ under the budget constraint
$$
 c_r+\sum_{i\in\mathcal{I}} c_i \ind_{\{ b_{ji}\neq 0\}} +b_{ji}\ind_{\{ b_{ji}=\inf B_i^*\}}\leq \bar{b}_j.
 $$

However, the  continuous bid space case provides different outcomes as explained below.

We investigate the smooth function case.
For each item $i,$ the random variable ${v}_{ji}$ has a continuously differentiable ($C^1$) cumulative distribution function $F_{ji}.$
Note  that
 the event $\{w: \ b_{ji}(v_j, w)=  b_{ki}(v_k, w)\}$ for $j\neq k$  is of measure zero since the valuation has a continuous probability density function.
Thus, the risk-sensitive LUBA analysis reduces to the (non-zero) lowest price on each item and their attitudes towards the risk.
\begin{lem}\label{lem1}
In the continuous bid space case,
the terms $P(x_i< \min_{j'\neq }b_{j'i})$  and $u(v_{ji}-x_i)P(x_i< \min_{j'\neq }b_{j'i})$ are decreasing in $x_i.$
\end{lem}
\begin{proof}It is because the cumulative distribution function $F(x)=P(X < x)$ increases with $x$  and the function $x \mapsto c-x$ deceases with $x.$ Observing  that $P(x_i< \min_{j'\neq }b_{j'i})$ can be written as $1-F(x_i),$ the assertion follows.
\end{proof}

Lemma \ref{lem1} provides a structural property of the better response bid in the continuous bid space case. The monotonicity properties above say that one can improve the payoff by  approaching the bid to zero. However, zero is not allowed bid.
When $j$ do not participate to the LUBA its payoff is $u^{-1}_j\left[ \mathbb{E} u_j(0,b_{-j})\right]=0.$ The risk-sensitive participation constraint yields $u^{-1}_j\left[ \mathbb{E} u_j(b)\right]\geq 0.$


\subsection{Bid Resubmissions}  Let $\mathcal{T}$ be the time space of one round of LUBA game, $\mathcal{T}=\{ 1,\ldots, T\}$ for $T\geq 1.$ At each time-step $t\in\mathcal{T}$ the bidders have opportunity to revise and resubmit another bid. A bidder has also the option not to resubmit and save the resubmission cost. A sequence of actions on item $i$ by bidder $j$  is $a_{ji}=(b_{ji}^1,\ldots, b_{ji}^T),$ where $b_{ji}^t$ can be a bid in $\delta \mathbb{N}$ or $NoBid="0".$

Note that each bidder can resubmit bids a certain number of times subject to her available budget, each resubmission for
 item $i$ will cost $c_i.$ Let $B_{ji}$ be the set of the non-zero (strictly positive) component of the action $a_{ji}.$ The cardinality of
 $B_{ji}$ is $n_{ji}=\sum_{t\in \mathcal{T}} \ind_{\{  b_{ji}^t> 0\}}.$
 If bidder $j$
 has (re)submitted $n_{ji}$ times on item $i$ her total submission/bidding cost would be $n_{ji}c_i$ in addition to the registration fee. Denote
 the set that contains all the bids of bidder $j$ on item $i$ by  $B_{ji}\subset \delta\mathbb{N}.$   The set of bidders who are submitting $b$ on item $i$ is denoted by
 $$ N_{i,b}=\{ j\in \mathcal{J} \ | \ b \in B_{ji} \}. $$
 The set of all positive natural numbers that were chosen by only one bidder on item $i$ is
 $B_{i}^*=\{ b>0 \ |  \  | N_{i,b} |=1 \}.$
 If $B_{i}^*=\emptyset$ then there is no winner on item $i$ at that round (after all the resubmission possibilities). If $B_{i}^*\neq \emptyset$ then there is a
 winner on item $i$ and the winning bid is $\inf B_{i}^*$ and winner is $j^*\in N_{i,\inf B_{i}^*}.$ The payoff of bidder $j$ on item $i$ at that round would be
 $
 r_{ji}=v_{ji}-|B_{ji}| c_i - \inf B_{i}^*-c_r,
 $ if $j$ is a winner on item $i$, and
 $
 r_{ji}=-|B_{ji}| c_i - c_r,
 $ if $j$ is not a winner on item $i.$   The payoff of bidder $j$ on item $i$ is zero if $B_{ji}$ is reduced to $\{0\}$ (or equivalently the empty set).

 \begin{eqnarray}&& r_{ji} (B) \\  \nonumber
 &=& [-c_r-c_i |B_{ji}|+(v_{ji}-b_{ji})\ind_{\{ b_{ji}=\inf B_i^*\}}]\ind_{\{  B_{ji}\neq \{0\}\}},
 \end{eqnarray}
where the infininum of the empty set is zero. The  payoff  is $r_j (B) =\sum_{i\in \mathcal{I}} r_{ji}(B)$ and the risk-sensitive payoff is  $u^{-1}_j\left[ \mathbb{E} u_j(r_j (B))\right].$ We choose the exponential function (with risk-sensitive index $\theta_j$ )to simplify the analysis.

\subsection{Risk-sensitive equilibria}

\begin{proposition}[Two bidders] Let $v>c+1.$
The risk-sensitive LUBA with resubmission has a mixed strategy given by
$$
(x_{0},x_{1},x_{12}, \ldots, x_{12\ldots,(k-1)},x_{12\ldots k})
$$
with $x_{12\ldots i}=\frac{e^{(i+1)\theta c}-e^{i\theta c}}{e^{\theta(v-(i+1))}-1}$,

$x_{12\ldots k}=1-x_0-\sum_{i=1}^{k-1}x_{12\ldots i}$ and 
$x_{12\ldots l}=0$ if $l>k$\\
where $u_j(r_j)=e^{\theta_jr_j}$.
\end{proposition}

\begin{proof}
Let $k$ be the largest integer such that $y_k=\mathbb{P}(\{0,1, . . . , k \}) > 0,$ $k\leq \bar{b}$. When  bidder 2's strategy is $y,$ the expected payoff of  sensitive bidder 1 when bidding $\{0,1,...,l\}$ is equal to 1 in equilibrium due to the indifference condition, for each $l \in \{1, 2,\ldots, k\}$. The cost of such a bid is equal to  $l c$.

On the other hand, the expected gain can be computed as follows. With probability $y_0$, bidder 2 will not post any bid, the winning bid is 1, and the gain is thus  $v - 1$ for any $l\leq  k.$
For $1\leq l\leq k,$ the bidder 2 bids $\{0,1,...,l\}$ with probability $y_i$ and the winner bid is $i+1$ from bidder 1, and 1's gain will $v-(l+1).$
Take the risk-sensitive perspective into consideration, The $\left[ \mathbb{E} u_j(r_j (B))\right]$ of bidder 1 when playing $\{0,...,l\}$ is therefore given by
$$\begin{array}{l}
\mbox{Action} \{0\}:\\
u(r_{1i}(\{0\}, y))=e^{\theta*0}=1\\
\mbox{Action} \{01\}:\\
u(r_{1i}(\{01\}, y))= e^{\theta(v-c-1)}y_0 +e^{-\theta c} y_1 +e^{-\theta c}(y_2+\ldots+y_k),\\
\mbox{Action} \{012\}:\\
u(r_{1i}(\{012\}, y))= e^{\theta(v-2c-1)}y_0 +e^{\theta(v-2c-2)}y_1\\ +e^{-\theta2c} y_2 +e^{-\theta 2c}(y_3+\ldots+y_k),\\
\ldots \\ \mbox{Action} \{012\ldots l\}:\\
u(r_{1i}(\{012\ldots l\}, y))= e^{\theta(v-lc-1)}y_0 \\ +e^{\theta(v-lc-2)}y_1\\ +\ldots +e^{\theta(v-lc-l)}y_{l-1}\\ +e^{-\theta lc} y_l +e^{-\theta lc}(y_{l+1}+\ldots+y_k)\\ \mbox{Action} \{012\ldots (l+1)\}:\\
u(r_{1i}(\{012\ldots l+1\}, y))= e^{\theta(v-(l+1)c-1)}y_0\\ +e^{\theta(v-(l+1)c-2)}y_1\\ +\ldots +e^{\theta(v-(l+1)c-(l+1))}y_{l}\\ +e^{-\theta(l+1)c} y_{l+1} +e^{-\theta(l+1)c}(y_{l+2}+\ldots+y_k)\\
\ldots \\ \mbox{Action} \{012\ldots k\}:\\
u(r_{1i}(\{012\ldots k\}, y))=  e^{\theta(v-kc-1)}y_0+e^{\theta(v-kc-2)}y_1\\ +\ldots+e^{\theta(v-kc-k+1)}y_{k-1}+e^{-\theta kc} y_k.\\
y_l\geq 0,\ y_0+\ldots+y_k=1,\\
 y_{k+s+1}=0 \mbox{ for}\ s\geq 0.
\end{array}$$

It turns out that

$$\left\{\begin{array}{c}
 (e^{\theta(v-1)}-1)y_0+1 =e^{\theta c}\\
 (e^{\theta (v-1)}-1)y_0 +(e^{\theta(v-2)}-1)y_1+1=e^{\theta 2c} \\
\ldots \\
(e^{\theta (v-1)}-1)y_0 +(e^{\theta (v-2)}-1)y_1+\ldots\\
 +(e^{\theta (v-l)}-1)y_{l-1}+1=e^{\theta lc} \\
 (e^{\theta (v-1)}-1)y_0 +(e^{\theta (v-2)}-1)y_1+\ldots \\
 +(e^{\theta (v-(l+1))}-1)y_{l}+1=e^{\theta (l+1)c} \\
\ldots \\
  (e^{\theta (v-1)}-1)y_0+(e^{\theta (v-2)}-1)y_1+\ldots\\
  +(e^{\theta (v-k)}-1)y_{k-1}+1=e^{\theta kc}.
\end{array} \right.
$$

For $l$ between $1$ and $k-1$ we make the difference between line $l+1$ and line $l$ to get:
$$\left\{
\begin{array}{c}
 y_0 = \frac{e^{\theta c}-1}{e^{\theta(v-1)}-1}\\
y_1=\frac{e^{2\theta c}-e^{\theta c}}{e^{\theta(v-2)}-1} \\
\ldots \\
y_{l-1}=\frac{e^{l\theta c}-e^{(l-1)\theta c}}{e^{\theta(v-l)}-1}  \\
y_{l}=\frac{e^{(l+1)\theta c}-e^{l\theta c}}{e^{\theta(v-(l+1))}-1} \\
\ldots \\
 y_{k-1}=\frac{e^{(k)\theta c}-e^{(k-1)\theta c}}{e^{\theta(v-k)}-1} \\
 y_k=1-(y_0+y_1+\ldots+y_{k-1}) >0\\
 y_{k+1+s}=0.
\end{array}
\right.
$$

Thus, the partially mixed strategy $$y^*=(\frac{e^{\theta c}-1}{e^{\theta(v-1)}-1},\frac{e^{2\theta c}-e^{\theta c}}{e^{\theta(v-2)}-1},\ldots, \frac{e^{(k)\theta c}-e^{(k-1)\theta c}}{e^{\theta(v-k)}-1},$$
 $$1-\sum_{l=0}^{k-1}\frac{e^{(l+1)\theta c}-e^{l\theta c}}{e^{\theta(v-l-1)}-1},0,\ldots, 0)$$ is an equilibrium strategy. The equilibrium sensitive-risk payoff is $\frac{1}{\theta_j}log{(\mathbb{E}e^{\theta_jr_j})}=0$.
\end{proof}
Note that, for two bidders the equilibrium strategy has monotone support. The set of actions $B_j$ such that $x_{B_j}>0$ is an increasing inclusion: $$\{ \} \subset \{ 1\}  \subset  \{ 12\} \subset  \{ 123\} \subset \{ 1234\}  \ldots  \{ 1234\ldots k\} , $$
it means that $ x_{2}= x_{3}=\ldots = x_{k}=0= x_{13}=x_{14}=x_{23}=x_{24}$ etc.

We analyze the Nash-Equilibrium of LUBA with resubmission in which the risk-sensitive payoffs are considered. Without losing generalization meaning, we specify $u_j(r_j)=e^{\theta_jr_j}$. So the risk-sensitive instant payoffs of bidder j is
$R_j=\frac{1}{\theta_j}log{(\mathbb{E}e^{\theta_jr_j})}$  where $\theta_j\neq 0$. For $\theta_j=0$, the $R_j$ is equal to $r_j$. The bidder $j$ can be divided into different categories based on $\theta_j.$
\begin{enumerate}
\item The bidder $j$ is risk-seeking for $\theta>0$
\item The bidder $j$ is risk neutrality for $\theta=0$
\item The bidder $j$ is risk aversion for $\theta<0$
\end{enumerate}
We analyze the simple case under the above assumption, where 2 bidders participate in the LUBA. We analyze different combinations of $\theta_j$ and give corresponding payoff matrix. In order to simplify the analysis, we assume that the bid space of each bidders is $\{0,1\}.$ Despite the simplification, this methodology can be extended to the general situation.

\begin{table}[ht]
\caption{$u_j(r_j)=e^{\theta_jr_j}$ of two risk sensitive bidders - One item}
\begin{center}
\begin{tabular}{c|c|c} \hline

 &    0                    & 1    \\ \hline
0& (1, 1)                  &(1, $e^{\theta_2(v-1)}$) \\ \hline
1&($e^{\theta_1(v-1)}$, 1) &($e^{-\theta_1}$, $e^{-\theta_2}$)\\\hline

\end{tabular}
\end{center}
\label{taby1}
\end{table}%

\begin{enumerate}
\item \textbf{Pure Nash-Equilibrium}: 
Table \ref{taby2} shows how the payoff matrix for risk-sensitive can be converted to the risk-neural payoff matrix  when analyzing pure Nash Equilibrium. So there are two pure Nash equilibrium.
\begin{table}[ht]
\caption{Two risk neutral bidders - One item}
\begin{center}
\begin{tabular}{c|c|c} \hline

 &    0                    & 1    \\ \hline
0& (0, 0)                  &$(0, v-1)^*$ \\ \hline
1&$(v-1, 0)^*$ & (-1,-1)\\\hline
\end{tabular}
\end{center}
\label{taby2}
\end{table}
\item \textbf{Mixed Nash-Equilibrium}: For bidder 2 the mixed strategy is:
\begin{enumerate}
\item Bidder 1: $(\frac{1-e^{\theta_2}}{1-e^{v\theta_2}},1-\frac{1-e^{\theta_2}}{1-e^{v\theta_2}})$
\item Bidder 2: $(\frac{1-e^{\theta_1}}{1-e^{v\theta_1}},1-\frac{1-e^{\theta_1}}{1-e^{v\theta_1}})$
\end{enumerate}
\end{enumerate}
\begin{figure}[!htb]
  \centering
    \includegraphics[width=0.9\textwidth]{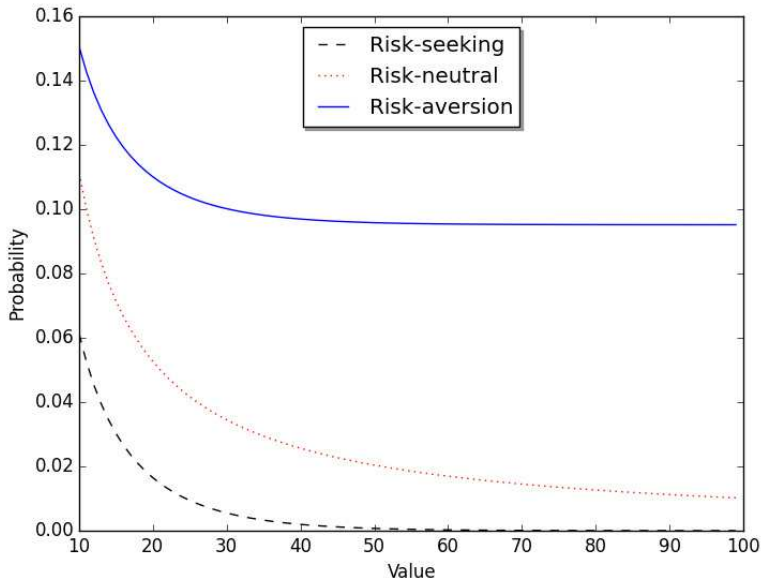}
  \caption{Probability of not participating in LUBA in mixed strategy for risk-sensitive bidders. The black dash line represents the strategy evolution with value for a risk-seeking bidder whose risk parameter is $\theta=0.2.$ The red dot line represents the strategy evolution with value for a risk-neutral bidder whose risk parameter is $\theta=0.$ The blue line is strategy evolution with value for a risk-aversion whose risk parameter is $\theta=-0.2.$} \label{simpleNash}
\end{figure}

As shown in Figure \ref{simpleNash}, the probability of participating in a risk-seeking bidder is greater than that of a risk-neutral one. And the probability of participating in a risk-aversion bidder is less than that of a risk-neutral one. This observation shows that over assumption on risk function is reasonable.

The next Theorem shows that this monotonicity property of the mixed equilibrium is lost when three or more bidders are involved.

\begin{proposition}[Three or more bidders] Suppose the realized parameters $(v,\theta)$ of the risk-sensitive LUBA are symmetric 
with three or more bidders. Then the risk-sensitive LUBA with resubmission has a  symmetric equilibrium in mixed strategies. However, 
there is no symmetric mixed equilibrium with monotone strategies.
\end{proposition}

The existence of symmetric mixed equilibrium in finite symmetric game is by now standard. We omit the details proofs of the non-monotonicity of the optimal strategies when three or more bidders are involved. The example below illustrates the result in the three-bidder case.

We analyze another simple example to show the properties of LUBA  with resubmission. We assume 3 bidders participate in a symmetric LUBA game. The submission fee is $c$, and the bid space is The specific parameters setting is €€$\{1,2,3\}$ and their combinations. 
The payoff matrix is shown in the Table \ref{tabm}.

\begin{table*}[ht]
\begin{tiny}
\caption{Payoff matrix in 3 bidders symmetric LUBA game}
\begin{center}
\begin{tabular}{|c|c|c|c|c|c|c|c|c|c|c|c|c|c|c|c|c|c|c}   \hline

  &   0,0    &   1,0     &2,0    &1.2,0   &0,1    &1,1    &2,1    &1.2,1     &0,2      &1,2      &2,2       &1.2  ,2   &0,1.2     &1,1.2     &2,1.2  &1.2,1.2                    \\ \hline
0& 0         &  0         &0       &0         &0      &0        &0       &0           &0         &0         &0         &0           &0           &0           &0         &0 \\ \hline
1&  v-c-1  & -c         &v-c-1 & -c       &-c     &-c       &-c      &-c           &v-c-1   &-c       &v-c-1     &-c        &-c      &-c          &-c        &-c\\ \hline
2& v-c-2  & -c         &-c      & -c       & -c    &v-c-2  &-c      &-c           &-c         &-c      &-c          &-c         &-c          &-c          &-c          &-c\\ \hline
1,2&v-2c-1&v-2c-2  &v-2c-1 &-2c    &v-2c-2&v-2c-2&-2c  &-2c        &v-2c-1   &-2c  &v-2c-1     &-2c      &-2c        &-2c         &-2c       &-2c\\ \hline
\end{tabular}
\end{center}
\label{tabm}
\end{tiny}
\end{table*}

Based on the payoff matrix, we introduce the following indifferent condition:
\begin{enumerate}
\item $\mbox{Payoff}(0)=0$
\item $\mbox{Payoff}(1)=(v-c-1)(y_0z_0+y_2z_0+y_0z_2+y_2z_2)-c[1-(y_0z_0+y_2z_0+y_0z_2+y_2z_2)]$
\item $\mbox{Payoff}(2)=(v-c-2)(y_0z_0+y_1z_1)-c[1-(y_0z_0+y_1z_1)]$
\item $\mbox{Payoff}(3)=(v-2c-1)Q_0+(v-2c-2)Q_1-2c[1-Q_0-Q_1]$\\
\end{enumerate}
where $Q_0=y_0z_0+y_2z_0+y_0z_2+y_2z_2$, and $Q_1=y_1z_0+y_1z_1+y_0z_1$.
As we assume that the proposed system is a symmetric LUBA game with three bidders, thus we can induce the following equations according the above indifferent condition equations.
\begin{enumerate}
\item $(v-c-1)(x_0^2+2x_0x_2+x_2^2)-c[1-(x_0^2+2x_0x_2+x_2^2)]=0$
\item $(v-c-2)(x_0^2+x_1^2)-c[1-(x_0^2+x_1^2)]=0$
\item $(v-2c-1)(x_0^2+2x_0x_2+x_2^2)+(v-2c-2)(2x_0x_1+x_1^2)-2c[1-(x_0^2+2x_0x_2+x_2^2)-(2x_0x_1+x_1^2)]$
\end{enumerate}
Then we can induce the following equations: $(x_0+x_2)^2=\frac{c}{v-1}$, $x_0^2+x_1^2=\frac{c}{v-2}$, and $2x_0x_1+x_1^2=\frac{c}{v-2}$ Combining the truth that $x_0+x_1+x_2+x_{1,2}=1$ we can derive the Nash-Equilibrium of this simple example is
 $\{x_0=2\sqrt{\frac{c}{5(v-2)}}, x_1=\sqrt{\frac{c}{5(v-2)}}, x_2=\sqrt{\frac{c}{v-1}}-2\sqrt{\frac{c}{5(v-2)}}, and  x_{\{1,2\}}=1-x_0-x_1-x_2\}$
We calculate four Nash-Equilibrium of different value, and the result is shown in Figure \ref{nashhello}.

\begin{figure}[!htb]
  \centering
    \includegraphics[width=0.9\textwidth]{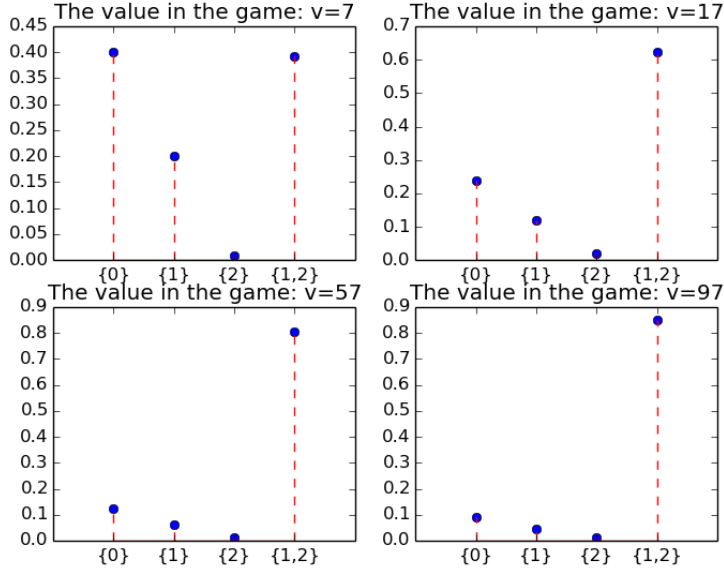}
  \caption{The Nash-Equilibrium evolves with the value $v$} \label{nashhello}
\end{figure}
We can observe that the probability of placing the combination bid ${1,2}$ is increasing with the value variable. Intuitively, it makes sense that when the valuation of the product is very high, the bidder is willing to take a risk for multi-bids. And when the value goes very high, the bid choice of $\{1,2\}$ dominates the whole bid choice space, as shown in the Figure \ref{nashhello}: 'the value in the game: $v=97$'.  

In a more general case, we show the evolution of Nash-Equilibrium with the value of the bidders targeting the item in Figure \ref{nashhello2}. By observing results, we can conclude that:
\begin{enumerate}
\item The probability of $\{0\}$ is higher than $\{1\}$ and $\{2\}$ in a Nash-Equilibrium.
\item The probabilities of single bid actions in the bid space decrease with the value in a Nash-Equilibrium
\item The probability of $\{1,2\}$ in a Nash-Equilibrium increases with the value.
\item The option of $\{1,2\}$ is dominant in action space when the valuation of bidders is very high. 
\end{enumerate}

\begin{figure}[!htb]
  \centering
    \includegraphics[width=0.9\textwidth]{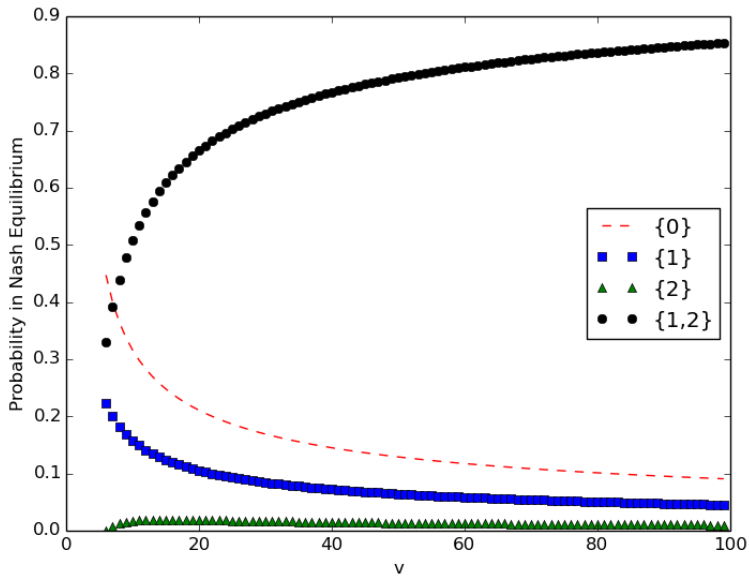}
  \caption{The Nash-Equilibrium evolves with the value $v$} \label{nashhello2}
\end{figure}

\subsection{Simulations of the Risk-Sensitive Scenarios}
\label{sec:4}
This section we will show and analyze the results from out proposed learning algorithm. In order to show the effectiveness of our Monte Carlo algorithm, we conduct a comparison experiment. In the comparison experiment, we compare the results of our proposed learning algorithm and our Monte Carlo algorithm under the same experiment parameters setting. The specific setting is shown in table \ref{tabelcom}.

\begin{table}[ht]
\caption{Parameters setting in the comparison experiment}
\begin{center}
\begin{tabular}{|c|c|c|c|}    \hline
 n &    10   &   m &    1               \\ \hline
 $\bar{b}$ &  10 & c &    1              \\ \hline
 $\alpha$ & 0.05 &$ \lambda $ &0.05      \\ \hline
\end{tabular}
\end{center}
\label{tabelcom}
\end{table}

We analyze the frequencies of each single bid from zero to nine, and the results is shown in Figure \ref{Fcompare}. As shown in Figure \ref{Fcompare}, the equilibria, learned by both of proposed algorithms, show the same tendency. Obviously, they are not exactly same. The reason of this phenomenon is that the Monto-Carlo algorithm is an approximate method and doesn't take consideration of value variable. However, the difference between the approximate result(obtained by Monte-Carlo algorithm) and the precise results(obtained by proposed learning algorithm) is very small. This phenomenon shows that our Monto-Carlo can effectively calculate the approximate Nash-Equilibrium in LUBA.
\begin{figure}[!htb]
  \centering
    \includegraphics[width=0.9\textwidth]{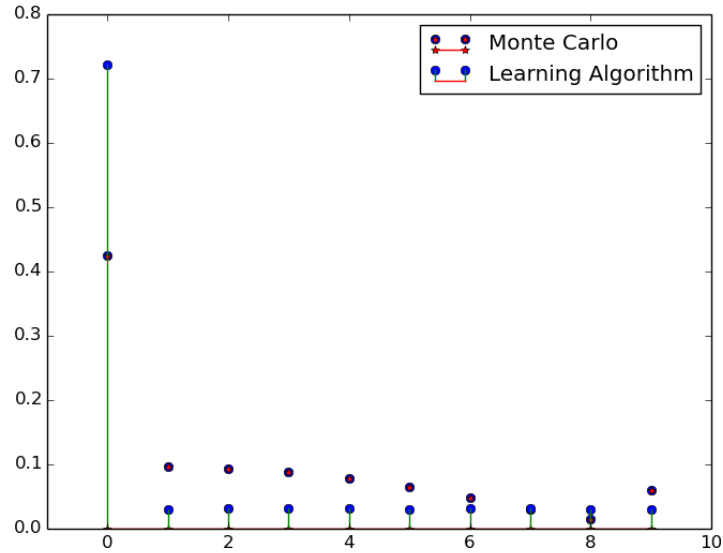}
     \caption{Strategy distribution on four different bidders of Item 1 after 5000 time iterations.}
  \label{Fcompare}
\end{figure}

Note that we show the effectiveness of our Monte-Carlo algorithm in small bid space, now we show the application of our Monte-Carlo algorithm in large bid space situation. We assume that 100 bidders participate the game, their budgets are sampled from 300 to 350 in a uniform distribution, and the value variable is from 90 to 100. The result is shown in Figure \ref{Fapproximate}.

\begin{figure}[!htb]
  \centering
    \includegraphics[width=0.9\textwidth]{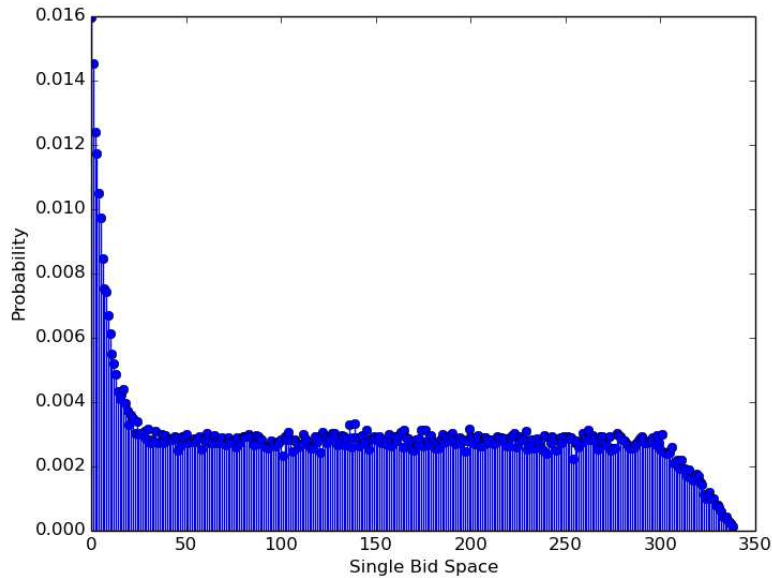}
     \caption{Strategy distribution on four different bidders of Item 1 after 5000 time iterations.}
  \label{Fapproximate}
\end{figure}

From the results we can draw the following conclusions. The tendency of approximated Nash Equilibrium is same as the small bid space situations. Practically, the strategy form like the tendency shown in Figure \ref{Fapproximate} is a suboptimal solution of approximating Nash Equilibrium. It makes sense for that if a bidder want to win the LUBA game, he or she has to block the small bid. And, in a multi bidders LUBA with resubmission, any relative big bids have a chance to win the product as the probability of placing a non-unique bid is high in this experiment setting. 
\section{Conclusion and Future Work} \label{sec:9}

In this paper, we have investigated multi-item lowest unique bid auctions in discrete bid spaces under heterogeneous budget constraints and incomplete information. Except for very special cases, the game does not  have equilibria in pure strategies. As a constrained finite game, there is at least one mixed Bayes-Nash equilibrium. A mixed equilibrium is explicitly computed in two bidder setup with resubmission possibilities. We have proposed a distributed strategic learning algorithm to approximate equilibria in the general setting. The numerical investigation has shown that the proposed algorithm can effectively learn Nash equilibrium  in few number of steps.  It is shown that the auctioneers can make a positive revenue when the number of bidders per bid exceeds a certain threshold.



\section*{Author Information}

{\it Yida Xu}  received  the B.Sc. degree in  Electronic Information Engineering from Chongqing University and the M.S. degree  in Information and Communication Engineering from Zhejiang University in China. He is a NYUAD Global Network Ph.D. candidate in the department of Electrical and Computer Engineering at New York University  Tandon School of Engineering. His research interests include auction theory, game theory and machine learning.

{\it Hamidou Tembine} (S'06-M'10-SM'13) received the M.S. degree in Applied Mathematics from Ecole Polytechnique in 2006 and the Ph.D. degree in Computer Science  from University of Avignon in 2009. His current research interests include evolutionary games, mean field stochastic games and applications. In December 2014, Tembine received the IEEE ComSoc Outstanding Young Researcher Award for his promising research activities for the benefit of the society. He was the recipient of 7 best article awards in the applications of game theory. Tembine is a prolific researcher and holds 150  scientific publications including magazines, letters, journals and conferences. He is author of the book on "distributed strategic learning for engineers" (published by CRC Press, Taylor \& Francis 2012), and co-author of the book "Game Theory and Learning in Wireless Networks" (Elsevier Academic Press). Tembine has been co-organizer of several scientific meetings on game theory in networking, wireless communications and smart energy systems.   He is a senior member of IEEE.

\end{document}